\documentclass[11pt]{article}

\usepackage[a4paper,top=100pt,bottom=100pt,left=30mm,right=30mm]{geometry}
\usepackage{enumitem}
\usepackage{multirow}

\setlist[enumerate]{itemsep=-2pt, topsep=5pt}

\usepackage{titlesec}
\titleformat{\section}[block]{\normalsize\bfseries\filcenter}{\thesection}{1em}{}
\titleformat{\subsection}[block]{\normalsize\scshape}{\thesubsection}{1em}{}

\usepackage[T1]{fontenc} 
\usepackage{textcomp}
\usepackage{times}
 \usepackage[scaled=0.92]{helvet} 

\usepackage{amsfonts}
\usepackage{amssymb, amsmath, amsthm}
\usepackage{graphicx}
\usepackage[mathscr]{eucal}

\usepackage{subfiles}
\usepackage{tikz}
\usetikzlibrary{shapes}
\usepackage{bm}
\usepackage{subcaption}
\usepackage{placeins}

\usepackage[colorinlistoftodos]{todonotes}
\usepackage[colorlinks=true, allcolors=blue]{hyperref} 

\usepackage{breakurl}

\DeclareMathOperator{\sgon}{sgon}
\DeclareMathOperator{\dgon}{dgon}
\DeclareMathOperator{\sdgon}{sdgon}
\DeclareMathOperator{\tw}{tw}

\DeclareMathOperator{\degree}{deg}

\newcommand{\N}{\mathbb{N}}
\newcommand{\Z}{\mathbb{Z}}

\theoremstyle{plain}
\newtheorem{theorem}{Theorem}[section]
\newtheorem{corollary}[theorem]{Corollary}
\newtheorem{lemma}[theorem]{Lemma}
\newtheorem{theoremintro}{Theorem}

\newtheorem{corintro}{Corollary}

\theoremstyle{definition}
\newtheorem{definition}[theorem]{Definition}
\newtheorem{example}[theorem]{Example}
\newtheorem{innercustomrule}{Rule}
\newenvironment{customrule}[1]
  {\renewcommand\theinnercustomrule{#1}\innercustomrule}
  {\endinnercustomrule}
\theoremstyle{remark}
\newtheorem{remark}[theorem]{Remark}

\definecolor{darkgreen}{RGB}{0,190,0}
\newcommand{\lus}[3]{
	\draw[#1] (#2,#3) to [relative, out=60,in=90] ({#2-.75},#3); 
	\draw[#1] (#2,#3) to [relative, out=-60,in=-90] ({#2-.75},#3);}
\newcommand{\cirkelsEnPijl}{
	\draw[gray] (1.5,0) circle (.75);
	\draw[gray] (4.5,0) circle (.75);
	\draw[->] (2.5,0) -- (3.25,0);
}
\newcommand{\alleenPijl}{
	\draw[->] (2.5,0) -- (3.25,0);
	\draw[white] (4.5,0) circle (.75);
}
\newcommand{\alleenPijlKlein}{
	\draw[->] (2.5,0) -- (3.25,0);
	\draw[white] (4.5,0) circle [x radius=.75, y radius=.3];
}

\tikzstyle{vertex}=[circle, draw, fill=black, inner sep=0pt, minimum width=3pt]
\tikzstyle{vertexx}=[circle, draw, fill=black, inner sep=0pt, minimum width=4pt]
\tikzstyle{added}=[diamond, fill=gray, inner sep=1pt, minimum width=4pt]
\tikzstyle{addedd}=[diamond, fill=gray, inner sep=1.25pt, minimum width=4pt]
\tikzstyle{groen} = [darkgreen, line width=1pt, dashed]
\tikzstyle{edge} = [line width = 1pt]
\tikzstyle{path} = [line width = 2pt]
\tikzstyle{rood} = [red, line width=1pt, dashed]

\begin{document}
\title{\textbf{Recognizing Hyperelliptic Graphs in Polynomial Time}\footnote{An extended abstract is published in \textit{Graph-Theoretic Concepts in Computer Science} \cite{BBCW}}}
\author{Jelco M.~Bodewes\thanks{Department Informatica, Universiteit Utrecht, Postbus 80.089, 3508 TB Utrecht,  Nederland, \texttt{jelcobodewes@gmail.com, h.l.bodlaender@uu.nl}} \\
Gunther Cornelissen\thanks{Mathematisch Instituut, Universiteit Utrecht, Postbus 80.010, 3508 TA Utrecht, Nederland, \texttt{g.cornelissen@uu.nl, m.vanderwegen@uu.nl}}
\and Hans L.~Bodlaender\footnotemark[2] \\Marieke van der Wegen\footnotemark[3]}
\date{}

\maketitle
\begin{abstract} 
Based on analogies between algebraic curves and graphs, Baker and Norine introduced \emph{divisorial gonality}, a graph parameter for multigraphs 
related to treewidth, multigraph algorithms and number theory. Various equivalent definitions of the gonality of an algebraic curve translate to \emph{different} notions of gonality for graphs, called \emph{stable gonality} and \emph{stable divisorial gonality}. 

We consider so-called {\em hyperelliptic graphs} (multigraphs of gonality $2$, in any meaning of graph gonality) and provide a safe and complete set of reduction rules for such multigraphs. This results in an algorithm to 
recognize hyperelliptic graphs in time $O(m+n \log n)$, where $n$ is the number of vertices and $m$ the number of edges of the multigraph. 
A corollary 
is that we can decide with the
same runtime whether a two-edge-connected graph $G$ admits an involution $\sigma$ such that the quotient $G/\langle \sigma \rangle$ is a tree.
\end{abstract}

\section{Introduction} \label{sec:intro}

\paragraph{Motivation}In this paper, we consider a graph theoretic problem that finds its origin in algebraic geometry, and can be formulated in terms of a specific type of graph search, namely {\em monotone chip firing}. The case with two chips is of special interest in the application, and we show that we can decide this case in $O(n \log n+m)$ time on a multigraph with $n$ vertices and $m$ edges.

In algebraic geometry, a special role is played by so-called \emph{hyperelliptic} curves; these are smooth projective algebraic curves possessing an involution, i.e.\ an automorphism of order two, for which the quotient is the projective line. Such curves can be described by an affine equation $y^2=f(x)$, for some one-variable polynomial $f(x)$ without repeated roots. They are widely studied and used, for example in the study of moduli spaces of abelian surfaces, invariants of binary quadratic forms, diophantine problems (finding integer or rational solutions to such equations), and in so-called hyperelliptic curve cryptography (see, e.g., \cite{Cohen} and \cite{Poonen}). 

Recognizing hyperelliptic curves is an important, decidable problem in algorithmic algebraic geometry; an algorithm has been implemented when the curve is given by some set of polynomial equations, e.g., in the computer algebra package \textsc{Magma} \cite{Magma}. No exact runtime analysis is available, but, the method being dependent on Gr\"obner basis computations, worst-case performance is expected to be more than exponential in the input size. 

In recent work of Baker and Norine \cite{Baker09}, the notion of a ``hyperelliptic graph'' was introduced, based on an analogy between algebraic curves and multigraphs. We show that the recognition problem for hyperelliptic graphs can be solved in quasilinear time. This can be applied to the recognition of certain hyperelliptic curves, since if an algebraic curve has a non-hyperelliptic stable reduction graph, the curve itself cannot be hyperelliptic (see \cite[3.5]{Baker}).

\paragraph{Divisorial gonality} Hyperelliptic graphs are graphs with \emph{divisorial gonality} at most two. The notion of divisorial gonality 
has several equivalent definitions; intuitively, we use a chip firing game: we have a graph and some initial configuration that assigns a non-negative number of ``chips'' to each vertex. We can fire a subset of vertices by moving a chip along each outgoing edge of the subset, \emph{if} every vertex has sufficiently many chips. We say that an initial configuration reaches a vertex if a sequence of firings results in that vertex having at least one chip. The divisorial gonality of a graph is the minimum number of chips needed for an initial configuration to reach each vertex
of the graph. It actually suffices to consider a `monotone' variant of the chip firing procedure, in which the sequence of subsets that
are fired to reach a vertex is increasing; 
this is similar to several other graph search games, where the optimal number of searchers does not increase when we
require the search to be monotone, see e.g., \cite{BienstockS91,LaPaugh93}.

\paragraph{Different notions of graph gonality} An equivalent definition of the gonality of an algebraic curve is the minimal degree of a morphism onto the projective line. A good analogue of this notion for graphs is the so-called \emph{stable gonality} introduced in \cite{C} as the minimal degree of a harmonic morphism from a refinement of the graph onto a tree (cf.\ Section \ref{subsec:sgon} infra). In contrast to the case of algebraic curves, the stable gonality of a graph is not always equal to its divisorial gonality. In \cite{C}, refinements of the graph arises from the theory of reductions of algebraic curves and tropical geometry \cite{Amini}, and it makes equal sense to consider the notion of \emph{stable divisorial gonality}, defined as the minimal divisorial gonality of a refinement of the graph.  

\paragraph{Known results} 
The termination of similar chip-firing games  
was discussed by Bj\"orner, Lov\'asz and Shor \cite{BLS}. A polynomial bound on the minimal number of required firings to terminate the Bj\"orner, Lov\'asz and Shor-game was given by Tardos \cite{tardos}. In the guise of ``abelian sandpile model'', chip-firing games play an important role in the study of self-organized criticality in statistical physics \cite{bak1988self,dhar1990self}. The chip firing game introduced by Baker and Norine is relevant for classical combinatorial problems about graphs, relating to spanning trees \cite{chan2015sandpiles}, the uniqueness of graph involutions \cite{Baker09}, and potential theory on electrical network graphs \cite{baker2013chip}. 

In \cite{Dutta}, a lower bound for the divisorial gonality of a graph is given in terms of its expansion. 
The gonality of a graph $G$ (in any sense) is larger than or equal to its treewidth $\tw(G)$ \cite{Gijs}. Since treewidth is insensitive to the presence of multiple edges while gonality is not, the parameters are different; actually, they are not ``tied'' in the sense of Norin \cite{Norin}; for example, there exists $G$ with $\tw(G)=2$ but $\dgon(G)$ arbitrarily high \cite{Hendrey}. The relation between the various notions of gonality is expounded in \cite[Section 1 and 5]{Gijs2}. 

We study the all three kinds of gonality of graphs from the point of view of computational complexity. The analogous problem of computing the gonality of an algebraic curve is decidable \cite{Schicho}. From the definition of divisorial gonality, it follows that divisorial gonality is computable. For stable gonality and stable divisorial gonality, this does not follow from the defintion, but both notions are computable as well \cite{GrootKoerkamp}, \cite{BWZ}. 
We know that treewidth is FPT, and that computing all three kinds of gonality is NP-hard and APX-hard \cite{Gijs2}. Moreover, divisorial gonality is in XP \cite[Section 5]{Bruyn2012reduced}. 

\paragraph{Our results} 
Our main result is the following.
\begin{theoremintro}[=Theorem \ref{thm:time}] \label{thm:hoofdstellingIntro}
There is an algorithm that decides whether a graph $G$ is hyperelliptic 
in $O(n\log n+m)$ time.
\end{theoremintro}
To obtain our algorithm, we 
provide a safe and complete set of reduction rules. We do this for all three notions of gonality. 
Similar to recognition algorithms for graphs of treewidth $2$ or $3$ (see \cite{Arnborg}), in our algorithm the rules are applied to the graph until no further 
rule application is possible; we decide positively if and only if this results in the empty graph.
One novelty is that some of the rules introduce constraints on pairs of vertices, which we model
by colored edges. To deal with the fact that some of the rules are not local, we use a data structure that allows us to find an efficient way of applying these rules, leading to the stated running time.

\paragraph{Application to detecting special involutions on graphs} There is no known polynomial time algorithm for the \textsc{graph automorphism problem}, the question whether a graph admits a non-trivial automorphism, a problem that is known to be in NP; recently, a quasi-polynomial time algorithm was given by Babai \cite{Babai} (compare \cite{Helfgott}). 

The question of the computational complexity of the problem is know to be very sensitive to alterations of the question. For example, deciding whether a graph has a \emph{fixed point free} automorphism of order two is NP-complete (see Lubiw \cite{Lubiw}). Our main result implies the following result as corollary. 
 
\begin{corintro}[=Corollary \ref{cor:automorfisme}] \label{cor:automorfismeIntro} There is an algorithm that, given a two-edge-connected graph $G$, decides in $O(n \log n + m)$ time whether $G$ admits an involution $\sigma$ such that the quotient $G/\langle \sigma \rangle$ is a tree.
\end{corintro}

\paragraph{Relation to number theory} We briefly elucidate the relevance of gonality for number theoretical problems. This paragraph can safely be skipped, but provides some motivation for the interest in computing gonality of graphs.  

If an algebraic curve $X$ is defined over the rational numbers and has gonality $\gamma$, then we have a so-called ``uniform boundedness'' result for $X$: the total number of points on $X$ with coordinates in any number field of degree $(\gamma-1)/2$, is finite. Now the gonality of $X$ is bounded below by the gonality of the dual graph of a reduction of the curve modulo a prime \cite[\S 11]{C}. We illustrate this with an example. 
\begin{figure}
\centering
\begin{subfigure}[b]{0.3\textwidth}
	\begin{tikzpicture}
	\draw (0,.5) -- (3,.5) (0,1.5) -- (3,1.5) (0,2) -- (3,2) (0,2.5) -- (3,2.5);
    \node (v) at (2,1.1) {$\ddots$};
    \draw (.5,0) -- (.5,3) (.5,0) -- (.5,3) (1,0) -- (1,3) (1.5,0) -- (1.5,3) (2.5,0) -- (2.5,3);
	\end{tikzpicture}
\end{subfigure} 
\hspace{0.1\textwidth}
\begin{subfigure}[b]{0.3\textwidth}
	\begin{tikzpicture}
	\node[vertex] (a1) at (0,2) {};
    \node[vertex] (a2) at (0,1.5) {};
    \node[vertex] (a3) at (0,1) {};
    \node (a4) at (0,.6) {$\vdots$}; 
    \node[vertex] (a5) at (0,0) {};
    \node[vertex] (b1) at (1.5,2) {};
    \node[vertex] (b2) at (1.5,1.5) {};
    \node[vertex] (b3) at (1.5,1) {};
    \node (b4) at (1.5,.6) {$\vdots$};
    \node[vertex] (b5) at (1.5,0) {};
    \draw (a1) -- (b1) (a1) -- (b2) (a1) -- (b3) (a1) -- (b5); 
    \draw (a2) -- (b1) (a2) -- (b2) (a2) -- (b3) (a2) -- (b5);
    \draw (a3) -- (b1) -- (a3) -- (b2) (a3) -- (b3) (a3) -- (b5); 
    \draw (a5) -- (b1) (a5) -- (b2) (a5) -- (b3) (a5) -- (b5);
    \draw[->] (2.5,1) -- (3,1);
    \node[vertex] (c1) at (3.5,2) {};
    \node[vertex] (c2) at (3.5,1.5) {};
    \node[vertex] (c3) at (3.5,1) {};
    \node (c4) at (3.5,.6) {$\vdots$}; 
    \node[vertex] (c5) at (3.5,0) {};
    \node[vertex] (d) at (4.5, 1) {};
    \draw (c1) -- (d) -- (c2) (c3) -- (d) -- (c5);
	\end{tikzpicture}
\end{subfigure}
\caption{The reduction of $X \colon (x^p-x)(y^p-y)=p$ modulo $p$ and the intersection dual graph $K_{p,p}$. }
\label{fig:voorbeeld}
\end{figure}
For a prime number $p$, consider the algebraic curve $X \colon (x^p-x)(y^p-y)=p$ in the $(x,y)$-plane over  the field of rational numbers $\mathbf Q$. Reducing the curve modulo $p$, the equation becomes  a union of lines $x(x-1)\dots(x-(p-1))\cdot y(y-1)\dots(y-(p-1))=0$  (see also Figure \ref{fig:voorbeeld}). The \emph{intersection dual graph} is given by  a vertex for each component of this reduction,  where two vertices are connected by an edge if and only if the corresponding components intersect; in the example, it is the complete bipartite graph $K_{p,p}$.  The stable gonality of $K_{p,p}$ is $p$ (since $\tw(K_{p,p}) = p$ and there is an obvious map of degree $p$ from $K_{p,p}$ to a tree, see Section \ref{subsec:sgon}). From \cite[4.5 \& 11.1]{C} one concludes that the set  $\bigcup X(K)$ is finite,  where $K$ runs over all the (infinitely many for $p \geq 5$) number fields of degree bounded above by $(p-1)/2$.

\section{Preliminaries} \label{sec:prelim}

Whenever we write ``graph'' we refer to a multigraph $G=(V,E)$, where $V$ is the set of vertices and $E$ is a multiset of edges. 

Let $G$ be a graph and $u$ and $v$ two vertices. Let $C$ be the connected component of $G \backslash \{v\}$ that contains $u$. By $G_v(u)$ we denote the induced subgraph of $G$ on $C \cup \{v\}$. 

\subsection{Divisorial Gonality}
\label{subsec:definitions}

There is a number of different definitions of divisorial gonality. The one we use is shown to be equivalent 
to the chip firing procedure without the `monotonicity' property by \cite{Bruyn2012reduced}. The definition given
here allows us to prove correctness of the reduction rules in our algorithm, and avoids more heavy algebraic
terminology.

A \emph{divisor} $D$ in a graph $G=(V,E)$ is a mapping $D\colon V \rightarrow \Z$ (a divisor represents a distribution of chips, see Section \ref{sec:intro}). We call a divisor $D$ \emph{effective} (notation $D\geq 0$) if $D(v) \geq 0$ for all $v\in V$. The degree, $\degree(D)$, of a divisor $D$ equals $\sum_{v\in V} D(v)$. 

Given an effective divisor $D$ and a set of vertices $W\subseteq V$, we call $W$ \emph{valid for $D$}, if for each $v\in W$, $D(v) \geq |E(v, V\setminus W)|$ (i.e., $v$ has at least as many chips as it has neighbors in $V\setminus W$). If $W$ is valid for $D$, we can \emph{fire} $W$ starting from $D$, this yields another divisor: for $v\in W$, $D(v)$ is decreased by the number of edges from $v$ to $V\setminus W$,
and for $x\in V\setminus W$, $D(x)$ is increased by the number of edges from $W$ to $x$. Intuitively, firing $W$ means moving a chip along all edges from $W$ to $V\setminus W$. Note that the divisor obtained by firing is effective as well. 

We call two effective divisors $D$ and $D'$ \emph{equivalent}, in notation $D\sim D'$, if there is a sequence of subsets $A_1 \subseteq A_2 \subseteq \ldots \subseteq A_{k-1} \subset A_k = V$, such that for all $i$ the set $A_i$ can be fired when $A_1, \ldots, A_{i-1}$ are fired starting from $D$, and the divisor obtained by firing $A_1, \ldots, A_k$ is $D'$. This defines an equivalence relation on the set of effective divisors \cite[Chapter 3]{Bruyn2012reduced}. For two equivalent effective divisors $D$ and $D'$, we call the difference of functions $D'-D$ the \emph{transformation} from $D$ to $D'$, and the sequence $A_1 \subseteq A_2 \subseteq \ldots \subseteq A_{k-1} \subset A_k = V$ the \emph{level set decomposition} of this transformation. This level set decomposition is unique \cite[Remark 3.8]{Bruyn2012reduced}.

We say that an effective divisor $D$ \emph{reaches} a vertex $v$, if there exists a $D'$ such that  $D\sim D'$ and $D'(v) \geq 1$. 
The \emph{divisorial gonality}, $\dgon(G)$, of a graph $G$ is the minimum degree of an effective divisor $D$ that reaches each vertex of $G$.

\begin{example} 
Let $T$ be a tree. Then $T$ has divisorial gonality 1. 
Let $v$ be a vertex of $T$ and consider the divisor $D$ with $D(v) = 1$ and $D(x) = 0$ for all $x\neq v$. This divisor has degree 1 and reaches each vertex of $T$: Let $w$ be a vertex of $T$. Let $vu$ be the first edge on the unique path from $v$ to $w$. Let $A_{v}$ be the component that contains $v$ of the cut induced by $vu$. Firing $A_{v}$ yields the divisor $D(u) = 1$ and $D(x) = 0$ for all $x\neq u$, thus we moved a chip from $v$ to $u$. Repeating this process yields a divisor with a chip on $w$. 
\end{example}

\begin{example}
Let $G$ be a cycle, then $G$ has divisorial gonality 2. First note that every set of vertices of $G$ induces a cut of size at least 2. Hence for all degree 1 divisors, there are no valid sets. Hence a degree 1 divisor does not reach every vertex. To see that there is a divisor with 2 chips that reaches every vertex, number the vertices $v_1, v_2, \ldots, v_n$ and consider the divisor $D$ with a chip on $v_1$ and a chip on $v_n$. To reach a vertex $v_k$ with $k \leq \frac{n}{2}$, fire the set $\{v_i \mid 1\leq i\leq j\} \cup \{v_i \mid n-j+1 \leq i \leq n\}$ for $j = 1, 2, \ldots, i-1$. Analogous for a vertex $v_k$ with $\frac{n}{2} \leq k \leq n$. 
\end{example}

\begin{example} \label{ex:treewidth2}
Consider the graph $G$ in Figure \ref{fig:treewidth2}. This graph has treewidth 1 and divisorial gonality 3. A divisor that reaches all vertices either has a chip on $u$ and 2 more chips to reach both $v$ and $w$, or has at least 3 chips to move along the three edges from $v$ to $u$. See also \cite[Table 3]{C}. 
\end{example}
\begin{figure}
\centering
	\begin{tikzpicture}
    \node[vertexx, label=$u$] (a) at (0,0) {};
    \node[vertexx, label=$v$] (b) at (2,0) {};
    \node[vertexx, label=$w$] (c) at (4,0) {};
    \draw[edge] (a) -- (b);
    \draw[edge] (a) to [relative, out=40, in=140] (b);
    \draw[edge] (b) to [relative, out=40, in=140] (a);
    \draw[edge] (b) to [relative, out=40, in=140] (c);
    \draw[edge] (c) to [relative, out=40, in=140] (b);
	\end{tikzpicture}
\caption{Graph with divisorial gonality 3 and treewidth 1 (see Example \ref{ex:treewidth2}). }
\label{fig:treewidth2}
\end{figure}

\begin{example} \label{ex:treewidth}
Consider the graph $G$ in Figure \ref{fig:treewidth}. This graph has treewidth 2 and divisorial gonality 3. A divisor that reaches all vertices needs two chips to traverse the left cycle and 2 chips to traverse the right cycle. But we cannot move two chips from $u$ to $v$, so these two chips on the left side cannot be the same as the two on the right side. Hence we need at least three chips.  
\end{example}
\begin{figure}
\centering
	\begin{tikzpicture}
    \node[vertexx, label=below:$u$] (a) at (.5,0) {};
    \node[vertexx] (b) at (-.25,.43) {};
    \node[vertexx] (c) at (-.25,-.43) {};
    \node[vertexx, label=below:$v$] (d) at (2.5,0) {};
    \node[vertexx] (e) at (3.25,.43) {};
    \node[vertexx] (f) at (3.25,-.43) {};
    \node[vertexx] (g) at (1.5,.5) {};
    \draw[edge] (a) -- (b) -- (c) -- (a);
    \draw[edge] (d) -- (e) -- (f) -- (d);
    \draw[edge] (d) -- (a) -- (g) -- (d);
	\end{tikzpicture}
\caption{Graph with divisorial gonality 3 and treewidth 2 (see Example \ref{ex:treewidth}). }
\label{fig:treewidth}
\end{figure}

For a disconnected graph, the divisorial gonality is equal to the sum of the divisorial gonalities of the connected components.

\subsection{Stable Gonality} \label{subsec:sgon}
We define stable gonality as in \cite[Definition 3.6]{C}. 

\begin{definition}
Let $G$ and $H$ be graphs. A \emph{finite morphism} is a map $\phi\colon G \to H$ such that \begin{enumerate}
\item[(i)] $\phi(V(G)) \subseteq V(H)$,
\item[(ii)] $\phi(uv) = \phi(u)\phi(v)$ for all $uv \in E(G)$, 
\end{enumerate}
together with, for every $e\in E(G)$, an ``index'' $r_\phi(e)\in \N$.
\end{definition}

\begin{definition}
We call a finite morphism $\phi\colon G\to H$ \emph{harmonic} if for every $v \in V(G)$ it holds that for all $e, e'\in E_{\phi(v)}(H)$
\begin{align*}
\sum_{d\in E_v(G), \phi(d) = e} r_\phi(d) = \sum_{d'\in E_v(G), \phi(d') = e'} r_\phi(d').
\end{align*}
We write $m_\phi(v)$ for this sum. 
\end{definition}

\begin{definition}
The \emph{degree} of a finite harmonic morphism $\phi\colon G\to H$ is 
\begin{align*}
\sum_{d\in E(G), \phi(d) = e} r_\phi(e) = \sum_{u\in V(G), \phi(u) = v} m_\phi(u),
\end{align*}
for $e\in E(H)$, $v\in V(H)$. This is independent of the choice of $e$ or $v$ (\cite{Baker06}, Lemma 2.4).
\end{definition}

\begin{example}\label{ex:tree1}
For a tree $T$ we can use the identity map $\phi\colon T \to T$, and assign index 1 to all edges, to obtain a finite harmonic morphism. This morphism has degree $1$. 
\end{example}

\begin{example} \label{ex:sgon-index1}
Consider the graph $G$ in Figure \ref{fig:sgon-index}. Assign index 2 to the edge $(v,w)$, and 1 to the other edges. Map this graph to a path on 4 vertices. This yields a finite harmonic morphism of degree $2$. 
\end{example}
\begin{figure}
\centering
	\begin{tikzpicture}
    \node[vertexx, label=$u$] (a) at (0,0) {};
    \node[vertexx, label=$v$] (b) at (2,0) {};
    \node[vertexx, label=$w$] (c) at (4,0) {};
    \node[vertexx, label=$x$] (d) at (6,0) {};
    \draw[edge] (a) to [relative, out=40, in=140] (b);
    \draw[edge] (b) to [relative, out=40, in=140] (a);
    \draw[edge] (d) to [relative, out=40, in=140] (c);
    \draw[edge] (c) to [relative, out=40, in=140] (d);
    \path[edge] (b) edge node[above] {$2$} (c);
    
    \draw[->] (3,-.4) -- (3,-1.1);
    
    \node[vertexx] (a) at (0,-1.5) {};
    \node[vertexx] (b) at (2,-1.5) {};
    \node[vertexx] (c) at (4,-1.5) {};
    \node[vertexx] (d) at (6,-1.5) {};
    \draw[edge] (a) -- (b) -- (c) -- (d);
	\end{tikzpicture}
\caption{Graph with stable gonality 2 (see Examples \ref{ex:sgon-index1} and \ref{ex:sgon-index2}). }
\label{fig:sgon-index}
\end{figure}

We can now proof a lemma about finite harmonic morphisms, that we will need in Section \ref{section:stablegonalityrules}. 

\begin{lemma} \label{lem:morphism-degree}
Let $G$ be a graph, and $\phi\colon G \to T$ a finite harmonic morphism of degree 2. If $\phi(u) = \phi(v)$, then $\deg(u) = \deg(v)$. 
\end{lemma} 

\begin{proof}
Notice that $m_\phi(u) = m_\phi(v) = 1$. Let $e$ be an edge incident to $\phi(u)$. By harmonicity, there is exactly one edge $e'$ such that $e'$ is incident to $u$ and $\phi(e') = e$. On the other hand every edge that is incident to $u$ is mapped to an edge that is incident to $\phi(u)$. So we conclude that $\deg_G(u) = \deg_T(\phi(u))$. Analogously we find that $\deg_G(v) = \deg_T(\phi(v))$. Since $\phi(u) = \phi(v)$, it follows that $\deg(u) = \deg(v)$.
\end{proof}

Before we can define the stable gonality of a graph, we need one last definition: the notion of refinements. 

\begin{definition}
A graph $G'$ is a \emph{refinement} of $G$ if $G'$ can be obtained by applying the following operations finitely many times to $G$.
\begin{enumerate}
\item[(i)] Add a leaf, i.e.\ a vertex of degree 1;
\item[(ii)] subdivide an edge by adding a vertex.
\end{enumerate}
We call a vertex of $G'\backslash G$ from which there are two disjoint paths to vertices of $G$, \emph{internal added} vertices, we call the other vertices of $G'\backslash G$ \emph{external added} vertices. 
\end{definition}

\begin{definition}
The \emph{stable gonality} of a graph $G$ is 
\begin{align*}
\sgon(G) = \min\{\deg(\phi) \mid \phi \colon  G' \to \ & T \text{ a finite harmonic morphism},\\
& G' \text{ a refinement of $G$, $T$ a tree}\}.
\end{align*}
\end{definition}

\begin{example} \label{ex:sgon-index2}
As we have seen in Example \ref{ex:tree1}, $\sgon(T) = 1$ for a tree $T$. On the other hand, if $G$ is not a tree, then any refinement of $G$ contains a cycle. Such a cycle cannot be mapped to a tree injectively. Thus $\sgon(G) > 1$ if $G$ is not a tree.

Since the graph $G$ in Figure \ref{fig:sgon-index} is not a tree and we have seen a morphism of degree $2$ in Example \ref{ex:sgon-index1}, we conclude that $\sgon(G) = 2$. 
\end{example}

\begin{example} \label{ex:sgon-dgon}
Consider the graph $G$ of Example \ref{ex:treewidth}, see Figure \ref{fig:treewidth}. This graph has stable gonality 2. Add a vertex to the edge $(u,v)$, a vertex to left triangle and a vertex to the right triangle. This refinement can be mapped to a path on 7 vertices, where $u$ is mapped to the third vertex and $v$ to the fifth vertex of the path.  If we assign index $1$ to all edges, this is a finite harmonic morphism of degree 2. 
\end{example}
\begin{figure}
\centering
	\begin{tikzpicture}
    \node[vertexx, label=below:$u$] (a) at (.5,0) {};
    \node[vertexx] (b) at (-.25,.43) {};
    \node[vertexx] (c) at (-.25,-.43) {};
    \node[vertexx, label=below:$v$] (d) at (2.5,0) {};
    \node[vertexx] (e) at (3.25,.43) {};
    \node[vertexx] (f) at (3.25,-.43) {};
    \node[vertexx] (g) at (1.5,.5) {};
    \node[addedd] (h) at (1.5,0) {};
    \node[addedd] (i) at (-.25,0) {};
    \node[addedd] (j) at (3.25,0) {};
    \draw[edge] (a) -- (b) -- (i) -- (c) -- (a);
    \draw[edge] (d) -- (e) -- (j) -- (f) -- (d);
    \draw[edge] (d) -- (h) -- (a) -- (g) -- (d);
    
    \draw[->] (1.5,-.8) -- (1.5,-1.5);
    
    \node[vertexx] (a) at (-.25,-2) {};
    \node[vertexx] (b) at (.5,-2) {};
    \node[vertexx] (c) at (1.5,-2) {};
    \node[vertexx] (d) at (2.5,-2) {};
    \node[vertexx] (e) at (3.25,-2) {};
    \node[addedd] (i) at (-.75,-2) {};
    \node[addedd] (j) at (3.75,-2) {};
    \draw[edge] (i) -- (a) -- (b) -- (c) -- (d) -- (e) -- (j);
	\end{tikzpicture}
\caption{Graph with stable gonality 3 and treewidth 2 (see Example \ref{ex:sgon-dgon}). }
\label{fig:sgon-dgon}
\end{figure}

For a disconnected graph $G$ its stable gonality is defined to be the sum of the stable gonalities of its components.

\subsection{Stable Divisorial Gonality} 

We can combine the previous two notion of gonality, first refine a graph and then consider the divisorial gonality, to obtain a third notion of gonality: stable divisorial gonality. 
\begin{definition}
The \emph{stable divisorial gonality} of $G$ is \begin{equation*}\sdgon(G) = \min\{\dgon(G') \mid G' \text{ a refinement of $G$}\}.\end{equation*}
\end{definition}

\begin{example} \label{ex:sdgon-dgon}
Consider the graph $G$ of Example \ref{ex:treewidth}, see Figure \ref{fig:treewidth}. This graph divisorial gonality $3$, but stable divisorial gonality 2. This is because we can refine $G$ to a graph with divisorial gonality 2: add a vertex to the edge $(u,v)$ (see Figure \ref{fig:sdgon-dgon}. A divisor with $2$ chips on vertex $u$ reaches all vertices. 
\end{example}
\begin{figure}
\centering
	\begin{tikzpicture}
    \node[vertexx, label=below:$u$] (a) at (.5,0) {};
    \node[vertexx] (b) at (-.25,.43) {};
    \node[vertexx] (c) at (-.25,-.43) {};
    \node[vertexx, label=below:$v$] (d) at (2.5,0) {};
    \node[vertexx] (e) at (3.25,.43) {};
    \node[vertexx] (f) at (3.25,-.43) {};
    \node[vertexx] (g) at (1.5,.5) {};
    \node[addedd] (h) at (1.5,0) {};
    \draw[edge] (a) -- (b) -- (c) -- (a);
    \draw[edge] (d) -- (e) -- (f) -- (d);
    \draw[edge] (d) -- (h) -- (a) -- (g) -- (d);
	\end{tikzpicture}
\caption{Graph with divisorial gonality 3 and stable divisorial gonality 2 (see Example \ref{ex:sdgon-dgon}). }
\label{fig:sdgon-dgon}
\end{figure}

For a disconnected graph $G$ its stable divisorial gonality is defined to be the sum of the stable divisorial gonalities of its components.

\subsection{Reduction Rules, Safeness and Completeness}

A \emph{reduction rule} is a rule that can be applied to a graph to produce a smaller graph. Our final goal with the set of reduction rules is to show that it can be used to characterize the graphs in a certain class, that of the graphs with divisorial gonality two, that of the graphs with stable gonality two, and that of graphs with stable divisorial gonality two, by reduction to the empty graph. For this we need to make sure that membership of the class is invariant under our reduction rules.

\begin{definition}
Let $\bm{U}$ be a rule and $\bm{S}$ be a set of reduction rules. Let $\mathcal{A}$ be a class of graphs. We call $\bm{U}$ \emph{safe} for $\mathcal{A}$ if for all graphs $G$ and $H$ such that $H$ can be produced by applying rule $\bm{U}$ to $G$ it follows that $H\in\mathcal{A} \Longleftrightarrow G\in\mathcal{A}$. We call $\bm{S}$ \emph{safe} for $\mathcal{A}$ if every rule in $\bm{S}$ is safe for $\mathcal{A}$.
\end{definition}

Apart from our rule sets being safe, we also need to know that, if a graph is in our class, it is always possible to reduce it to the empty graph.

\begin{definition}
Let $\bm{S}$ be a set of reduction rules and $\mathcal{A}$ be a class of graphs. We call $\bm{S}$ \emph{complete} for $\mathcal{A}$ if for any graph $G\in\mathcal{A}$ it holds that $G$ can be reduced to the empty graph by applying some finite sequence of rules from $\bm{S}$.
\end{definition}

For any rule set that is both complete and safe for $\mathcal{A}$ the rule set is suitable for characterizing  $\mathcal{A}$: a graph $G$ can be reduced to the empty graph if and only if $G$ is in $\mathcal{A}$.  Additionally it is not possible to make a wrong choice early on that would prevent the graph from being reduced to the empty graph: if $G\in \mathcal{A}$ and $G$ can be reduced to $H$, then $H$ can be reduced to the empty graph.

These properties ensure that we can use the set of reduction rules to create an algorithm for recognition of the graph class.

\subsection{Constraints}
In the process of applying reduction rules to a graph, we will need to keep track of certain restrictions otherwise lost by removal of vertices and edges. We will maintain these restrictions in the form of a set of pairs of vertices, called constraints, and then extend the notions of gonality to graphs with constraints.

\begin{definition}
Given a graph $G=(V,E)$, a \emph{constraint} on $G$ is an unordered pair of vertices $v,w\in V$, usually denoted as $(v,w)$, where $v$ and $w$ can be the same vertex.
\end{definition}

A graph with contraints consists of a graph $G=(V,E)$ and a set of constraints $\mathscr{C}$. 
Constraints are, like edges, pairs of vertices, so we can consider them as an extra set of edges. The conditions that a constaint places on the divisors and firing sets for divisorial gonality and the morphisms for stable gonality, are described in Sections \ref{divSection} and \ref{section:stablegonalityrules}, respectively.

\section{Reduction Rules for Divisorial Gonality} 
\label{divSection}

We will now show that there exists a set of reduction rules that is safe and complete for the class of graphs with divisorial gonality at most two.
We will assume that our graph is loopless and connected. Loops can simply be removed from the graph since they never impact the divisorial gonality and a disconnected graph has divisorial gonality two or lower exactly when it consists of two trees, which can easily be checked in linear time.

\subsection*{Constraints for Divisorial Gonality}

Checking whether a graph has gonality two or lower is the same as
checking whether there exists a divisor on our graph with degree two that reaches all vertices. Our constraints place restrictions on what divisors we consider, as well as what sets we are allowed to fire. 

\begin{definition}
Given a graph $G$ with set of constraints $\mathscr{C}$, and two equivalent effective divisors $D$ and $D'$. We call $D$ and $D'$ \emph{$\mathscr{C}$-equivalent} (in notation $D \sim_{\mathscr{C}} D'$), if for every set $A_i$ of the level set decomposition of $D' - D$ and every constraint $(u,v) \in \mathscr{C}$, either $u,v\in A_i$ or $u,v\notin A_i$. 
\end{definition}
Note that this defines a finer equivalence relation. Now we can extend the definition of \emph{reach} using $\mathscr{C}$-equivalence: a divisor $D$ \emph{reaches} a vertex $v$, if there exists a $D'$ such that  $D\sim_\mathscr{C} D'$ and $D'(v) \geq 1$.
\begin{definition}
Given a graph $G$ with a set of constraints $\mathscr{C}$. A divisor $D$ \emph{satisfies} $\mathscr{C}$ if for every constraint $(u,v) \in \mathscr{C}$ there is a divisor $D' \sim_{\mathscr{C}} D$ such that $D'(u)\geq 1$ and $D'(v) \geq 1$ if $u \neq v$ and $D'(u)\geq 2$ if $u=v$.  
\end{definition}

\begin{definition}
Given a graph $G=(V,E)$ with constraints $\mathscr{C}$, we call a divisor $D$ \emph{suitable} if it is effective, has degree $2$, reaches all vertices using the $\mathscr{C}$-equivalence relation and satisfies all constraints in $\mathscr{C}$. 
\end{definition}

\begin{definition} We will say that a graph with constraints has divisorial gonality $2$ or lower if it admits a suitable divisor. Note that for a graph with no constraints this is equivalent to the usual definition of divisorial gonality $2$ or lower. We will denote the class of graphs with constraints that have divisorial gonality two or lower as $\mathcal{G}^d_2$.  
\end{definition}

\paragraph{Constraints \& Cycles} It will be useful to determine when constraints are non-conflicting locally:
  
\begin{definition}
Let $C$ be a cycle in a graph $G$ with constraints $\mathscr{C}$. Let $\mathscr{C}_C\subseteq \mathscr{C}$ be the subset of the constraints that contain a vertex in $C$. We call the constraints $\mathscr{C}_C$ \emph{compatible} if the following hold. 
\begin{enumerate}
\item[\textup{(i)}] If $(v,w)\in \mathscr{C}_C$ then both $v\in C$ and $w\in C$. 
\item[\textup{(ii)}] For each $(v,w)\in \mathscr{C}_C$ and $(v', w')\in \mathscr{C}_C$, the divisor given by assigning a chip to $v$ and $w$ must be equivalent to the one given by assigning a chip to $v'$ and $w'$ on the subgraph consisting of $C$.
\end{enumerate}
\end{definition}

\subsection*{The Reduction Rules}
We are given a connected loopless graph $G=(V,E)$ and a yet empty set of constraints $\mathscr{C}$. The following rules are illustrated in Figure \ref{fig:dgonRules}, where a constraint is represented by a red dashed edge. 

We start by covering the two possible end states of our reduction:

\begin{customrule}{$\bm{E_1^d}$}\label{rul:de1}
Given a graph consisting of exactly one vertex, remove that vertex.
\end{customrule}
\begin{customrule}{$\bm{E_2^d}$}\label{rul:de2}
Given a graph consisting of exactly two vertices, $u$ and $v$, connected to each other by a single edge, and ${\mathscr{C}}=\{(u,v)\}$, remove both vertices.
\end{customrule}

Next are the reduction rules to get rid of vertices with degree one. These rules are split by what constraint applies to the vertex:

\begin{customrule}{$\bm{T_1^d}$} \label{rul:dt1}
Let $v$ be a leaf, such that $v$ has no constraints in $\mathscr{C}$. Remove $v$.
\end{customrule}
\begin{customrule}{$\bm{T_2^d}$} \label{rul:dt2}
Let $v$ be a leaf, such that its only constraint in $\mathscr{C}$ is $(v,v)$. Let $u$ be its neighbor. Remove $v$ and add the constraint $(u,u)$ if it does not exist yet.
\end{customrule}
\begin{customrule}{$\bm{T_3^d}$} \label{rul:dt3}
Let $v_1$ be a leaf, such that its only constraint in $\mathscr{C}$ is $(v_1,v_2)$, where $v_2$ is another leaf, whose only constraint is also $(v_1,v_2)$. Let $u_1$ be the neighbor of $v_1$ and $u_2$ be the neighbor of $v_2$ (these can be the same vertex). Then remove $v_1$ and $v_2$ and add the constraint $(u_1, u_2)$ if it does not exist yet. 
\end{customrule}

Finally we have a set of reduction rules that apply to cycles containing at most $2$ vertices with degree greater than two. The rules themselves are split by the number of vertices with degree greater than two.

\begin{customrule}{$\bm{C_1^d}$} \label{rul:dc1}
Let $C$ be a cycle of vertices with degree two. If the set of constraints $\mathscr{C}_C$ on $C$ is compatible, then replace $C$ by a new single vertex. 
\end{customrule}
\begin{customrule}{$\bm{C_2^d}$} \label{rul:dc2}
Let $C$ be a cycle with one vertex $v$ with degree greater than two. If the set of constraints $\mathscr{C}_C$ on $C$ plus the constraint $(v,v)$ is compatible, then remove all vertices except $v$ in $C$ and add the constraint $(v,v)$ if it does not exist yet.
\end{customrule}
\begin{customrule}{$\bm{C_3^d}$} \label{rul:dc3}
Let $C$ be a cycle with two vertices $v$ and $u$ of degree greater than two. If there exists a path from $v$ to $u$ that does not share any edges with $C$ and the set of constraints $\mathscr{C}_C$ on $C$ plus the constraint $(v,u)$ is compatible, then remove all vertices of $C$ except $v$ and $u$, remove all edges in $C$ and add the constraint $(v,u)$ if it does not exist yet.
\end{customrule}

\begin{figure}[t]
		\centering
		\begin{tabular}{|rlr}
        	\hline
			\multicolumn{1}{|l|}{\small{Rule $\bm{E_1^d}$}} & & \multicolumn{1}{l|}{\small{Rule $\bm{E_2^d}$}} \\
			\multicolumn{1}{|r|}{\begin{tikzpicture}
	\alleenPijl
	\node[vertexx] (u) at (1.25,0) {};
    \draw (4.25,0) circle (0.2);
    \draw (4.45,0.2) -- (4.05, -0.2);
\end{tikzpicture} } & & \multicolumn{1}{r|}{\begin{tikzpicture}
	\alleenPijl
	\node[vertexx] (u) at (1.25,0) {};
	\node[vertexx] (v) at (0,0) {};
	\draw[edge] (u)--(v);
    \draw (4.25,0) circle (0.2);
    \draw (4.45,0.2) -- (4.05, -0.2);
	\draw[rood] (u) to [relative, out=310, in=220] (v);
\end{tikzpicture} } \\
        	\hline
            \multicolumn{1}{|l|}{\small{Rule $\bm{T_1^d}$}} & & \multicolumn{1}{l|}{\small{Rule $\bm{T_2^d}$}} \\
            \multicolumn{1}{|r|}{\begin{tikzpicture}
	\cirkelsEnPijl
	\node[vertexx] (u) at (1.25,0) {};
	\node[vertexx] (v) at (0,0) {};
	\draw[edge] (u)--(v);
	\node[vertexx] (u') at (4.25,0) {};
\end{tikzpicture} 
			} & & \multicolumn{1}{r|}{\begin{tikzpicture}
	\cirkelsEnPijl
	\lus{rood}{0}{0}
	\lus{rood}{4.25}{0}
	\node[vertexx] (u) at (1.25,0) {};
	\node[vertexx] (v) at (0,0) {};
	\node[vertexx] (u') at (4.25,0) {};
	\draw[edge] (u)--(v);
\end{tikzpicture} } \\
        	\hline
            \multicolumn{1}{|l}{\small{Rule $\bm{T_3^d}$}} & & \multicolumn{1}{l|}{}  \\
            \multicolumn{1}{|r}{\begin{tikzpicture}
	\cirkelsEnPijl
	\lus{rood}{4.25}{0}
	\node[vertexx] (w) at (1.25,0) {};
	\node[vertexx] (u) at (0,.375) {};
	\node[vertexx] (v) at (0,-.375) {};
	\draw[edge] (u)--(w) -- (v);
	\draw[rood] (u) --(v);
	\node[vertexx] (u') at (4.25,0) {};
\end{tikzpicture} } & & \multicolumn{1}{r|}{ \begin{tikzpicture}
	\cirkelsEnPijl
	\node[vertexx] (w1) at (1.25,.375) {};
	\node[vertexx] (w2) at (1.25,-.375) {};
	\node[vertexx] (u) at (0,.375) {};
	\node[vertexx] (v) at (0,-.375) {};
	\draw[edge] (u)--(w1) (w2) -- (v);
	\draw[rood] (u) --(v);
	\node[vertexx] (u') at (4.25,0.375) {};
	\node[vertexx] (v') at (4.25,-.375) {};
	\draw[rood] (u') -- (v');
\end{tikzpicture} } \\
        	\hline
            \multicolumn{1}{|l|}{\small{Rule $\bm{C_1^d}$}} & & \multicolumn{1}{l|}{\small{Rule $\bm{C_2^d}$}} \\
            \multicolumn{1}{|r|}{\begin{tikzpicture}
	\alleenPijl
    \draw[path] (1,0) circle (0.6);
	\node[vertexx] (u) at (4.25,0) {};
\end{tikzpicture} } & & \multicolumn{1}{r|}{\begin{tikzpicture}
	\cirkelsEnPijl
	\draw[path] (0.75,0) circle (0.5);
	\lus{rood}{4.25}{0}
	\node[vertexx] (u) at (1.25,0) {};
	\node[vertexx] (u') at (4.25,0) {};
\end{tikzpicture} } \\
        	\hline
            \multicolumn{1}{|l|}{\small{Rule $\bm{C_3^d}$}} & &  \\
            \multicolumn{1}{|r|}{\begin{tikzpicture}
	\cirkelsEnPijl
	\node[vertexx] (u) at (1.25,.375) {};
	\node[vertexx] (v) at (1.25,-.375) {};
	\node[vertexx] (u') at (4.25,0.375) {};
	\node[vertexx] (v') at (4.25,-.375) {};
	\draw[path] (1.25,0) circle (0.375);
	
	\draw[dotted, thick, rounded corners=10] (v) -- (2,-.5) -- (1.75,.5)  -- (u);
	\draw[rood] (u') -- (v');
	\draw[dotted, thick, rounded corners=10] (v') -- (5,-.5) -- (4.75,.5)  -- (u');
\end{tikzpicture} } & &        \\
        	\cline{1-1}
		\end{tabular}
		\caption{The reduction rules for divisorial gonality}  
		\label{fig:dgonRules}
	\end{figure}

We denote by $\mathcal{R}^d$ the set consisting of all the above reduction rules: $\bm{E_1^d}$, $\bm{E_2^d}$, $\bm{T_1^d}$, $\bm{T_2^d}$, $\bm{T_3^d}$, $\bm{C_1^d}$, $\bm{C_2^d}$ and $\bm{C_3^d}$. 

We will now state the main theorem stating that this set of reduction rules has the desired properties. After this we will build up the proof.

\begin{theorem} \label{mainTheorem}
The set of rules $\mathcal{R}^d$ is safe and complete for $\mathcal{G}_2^d$.
\end{theorem}

\subsection*{Safeness}
In this section it is assumed there is a graph $G$ and another graph $H$ that follows from $G$ by applying a rule. Now we first make an observation on the connectivity of our graphs:

\begin{lemma} \label{connectedLemma}
Let $G$ and $H$ be graphs. If $G$ is connected and $H$ can be produced from $G$ by applying some rules, then $H$ is connected.
\end{lemma}
\begin{proof}
We observe that the only rule that removes a path between two remaining vertices is \ref{rul:dc3}. In the case of \ref{rul:dc3} however we demand that there is a path between $v$ and $w$ outside of $C$ so this path will still exist and it follows that $H$ is still connected.
\end{proof}

Since we assume our graph $G$ is connected it follows that each produced graph $H$ is also connected. Now we will show for each of the rules in $\mathcal{R}^d$ that it is safe.

\begin{lemma}
Rules \ref{rul:de1} and \ref{rul:de2} are safe.
\end{lemma}
\begin{proof}
For both rules it should be clear that their starting states as well as the empty graph have divisorial gonality two or lower. From this it follows they both are safe.
\end{proof}

\begin{lemma}
Rules \ref{rul:dt1} and \ref{rul:dt2} are safe.
\end{lemma}
\begin{proof}
Let $v$ be our vertex with degree $1$ and $u$ its neighbor. We know that the only constraint on $v$ can be the constraint $(v,v)$. 

Note that if $H\in \mathcal{G}_2^d$ then there is a divisor on $H$ that puts at least one chip on $u$. Considering this divisor on $G$, note that we can move chips to $v$ by firing $G-\{v\}$, it follows that this divisor is also suitable for $G$.

Given that $G\in \mathcal{G}_2^d$ note that we can find a suitable divisor that has no chips on $v$ by firing $v$ until it contains no chips. This divisor will also be suitable on $H$.

For \ref{rul:dt2} the proof is analogous, except with two chips on $v$.
\end{proof}

\begin{lemma}
Rule \ref{rul:dt3} is safe.
\end{lemma}
\begin{proof}
Let $v_1$ and $v_2$ be the vertices with degree one, such that their only constraint is $(v_1,v_2)$ and let $u_1$ and $u_2$ be their  (possibly equal) neighbors. We first assume that $H\in \mathcal{G}_2^d$, then there is a suitable divisor on $H$ with one chip on $u_1$ and another chip on $u_2$. Consider this divisor on $G$. Then by firing $V(G)\setminus\{ v_1, v_2\} $ we can move a chip to $v_1$ and $v_2$. For every vertex $v\in V(G)\setminus \{v_1,v_2\}$ there is a sequence $A_1, A_2, \ldots, A_{k} \subseteq V(H)$ such that firing this sequence yields a divisor $D'$ with a chip on $v$. Now add $v_i$ to every set $A_j$ that contains $u_i$ for $i=1,2$. Firing these sets on $G$ starting from $D$ results in $D'$ on $G$, so $D$ reaches $v$. Moreover, every set we fired contains either both $v_1$ and $v_2$, or neither. We conclude that $D$ is also suitable on $G$. 

Assume that $G\in \mathcal{G}_2^d$, then the divisor on $G$ with one chip on $v_1$ and $v_2$ is suitable. By firing $\{ v_1, v_2\}$ we can create a divisor with a chip on $u_1$ and $u_2$ (or two on $u_1$ if $u_1=u_2$). It follows that this divisor is suitable when considered on $H$. 
\end{proof}

\begin{lemma}
Rule \ref{rul:dc1} is safe.
\end{lemma}
\begin{proof}
We start by assuming that $H\in \mathcal{G}_2^d$. Note that by Lemma \ref{connectedLemma} we have that $H$ is connected. Therefore it follows that $H$ must consist of a single vertex, and $G$ consists of a single cycle. It follows that $G\in \mathcal{G}_2^d$, since all constraints are compatible. 

Assume then that $G\in \mathcal{G}_2^d$ instead. Since $G$ is connected it must consist exactly of the cycle $C$, thus $H$ consists of a single point and $H\in \mathcal{G}_2^d$. 
\end{proof}

\begin{lemma}
Rule \ref{rul:dc2} is safe.
\end{lemma}
\begin{proof}
Let $C$ be our cycle with one vertex $v$ with degree greater than $2$. Assume that $H\in \mathcal{G}_2^d$; then there is a suitable divisor on $H$ with two chips on $v$. Consider this divisor on $G$. Note that if we fire $V(G)-C+\{ v\}$ then we move the two chips onto the two neighbors of $v$ in $C$. Since all constraints on $C$ are compatible with the constraint $(v,v)$ it follows that we can move the chips along $C$ while satisfying the constraints on $C$. From this it follows that our divisor is suitable on $G$.

Assume now that $G\in \mathcal{G}_2^d$. Since all constraints on $C$ are compatible with $(v,v)$, it follows that we can find a suitable divisor with two chips on $v$. Considering this divisor on $H$ gives a suitable divisor there. Thus, $H\in \mathcal{G}_2^d$. 
\end{proof}

\begin{lemma} \label{lem:ruleC3safe}
Rule \ref{rul:dc3} is safe.
\end{lemma}
\begin{proof}
Let $C$ be our cycle and $v,w$ the two vertices with degree greater than two in $C$. We first assume that $H\in \mathcal{G}_2^d$. From this it follows that the divisor on $H$ with a chip on $v$ and a chip on $w$ is suitable. Consider this divisor on $G$. It is clear that from $v$ and $w$ we can move chips along either of the two arcs between $v$ and $w$ in $C$. We know that in $G$ all constraints on $C$ plus $(v,w)$ are compatible. Therefore the divisor is also suitable on $G$ and thus $G\in \mathcal{G}_2^d$. 

Let us now assume that instead $G\in \mathcal{G}_2^d$. Clearly there exists a suitable divisor $D$ on $G$ that has a chip on $v$. We will show that there is a suitable divisor that has a chip on both $v$ and $w$: Assume that $D(w)=0$, then there should be a suitable divisor $D'$ with $D'(w)=1$ and $D\sim_{\mathscr{C}} D'$. This implies there is a level set decomposition $A_1, \ldots, A_k$ of the transformation from $D$ to $D'$.

Let $A_i$ be the first subset that contains $v$ and $D_i$ the divisor before firing $A_i$. Notice that $D_i(v)\geq 1$, since $D(v) \geq 0$ and $v$ is not fired yet.  Notice that we have $D_i(a)\geq |E(a, V(G)\setminus A_i)|$ for all $a\in A_i$, since all firing sets are valid. Since $\degree (D_i)=2$  it follows that $\sum_{a\in A_i} |E(a, V(G)\setminus A_i)| \leq 2$. This is the same as the cut induced by $A_i$ having size two or lower. The minimum cut between $v$ and $w$ is at least three, since they are both part of $C$ and there exists an additional path outside of $C$ between them. Therefore it follows that $A_i$ can only induce a cut of size two or lower if $w\in A_i$ as well. But this implies that $D_i(w)\geq 1$, since a vertex can not receive a chip after entering the firing set. We conclude that $D_i(v)=1$ and $D_i(w)=1$.

Also by the fact that the minimum cut between $v$ and $w$ is at least three it follows that a subset firing can only be valid if the subset contains either both $v$ and $w$ or neither (since otherwise the subset would have at least three outgoing edges). It follows we can satisfy the set of constraints including $(v,w)$.

Therefore the divisor $D_i$ gives us a suitable divisor when considered on $H$. We conclude that $H\in \mathcal{G}_2^d$.
\end{proof}

Since we have shown that each of the rules in $\mathcal{R}^d$ is safe, we conclude:

\begin{lemma} \label{lem:divSafe}
The ruleset $\mathcal{R}^d$ is safe for $\mathcal{G}_2^d$. \qed 
\end{lemma}

\subsection*{Completeness}
By Lemma \ref{lem:divSafe} we have that membership in $\mathcal{G}_2^d$ is invariant under the reduction rules in $\mathcal{R}^d$. For the reduction rules to be useful however we will also need to confirm that any graph can be reduced to the empty graph by a finite sequence of rule applications.

\begin{lemma} \label{ConflictingLabelsLemma}
Let $G$ be a graph and $v\in V(G)$ a vertex. If there are two different constraints on $v$, so $(v,w), (v,w')\in \mathscr{C}$, with $w\neq w'$, then $G\notin \mathcal{G}_2^d$.
\end{lemma}

\begin{proof}
Any suitable divisor should be equivalent to the divisor $D$ with $D(v)=1$, $D(w)=1$ (or $D(v) = 2$ is $w=v$), and equivalent to the divisor $D'$ with $D'(v)=1$, $D'(w')=1$. This means that these divisors are equivalent to each other. Notice that any firing set that contains $v$ also contains both $w$ and $w'$ by our constraints.  Moreover, any firing set containing $w$ contains $v$ and $w'$ by our constraints. Starting with divisor $D$, notice that any valid firing set must contain $v$ or $w$ (they are the only vertices with chips). It follows that it also contains $w'$. This implies that the number of chips on $w'$ cannot increase, so no level set decomposition from $D$ to $D'$ exists, thus $D$ and $D'$ cannot be equivalent. We conclude no suitable divisor exists and therefore $G\notin \mathcal{G}_2^d$.
\end{proof}

\begin{lemma} \label{DegreeOneLemma}
Let $G\in \mathcal{G}_2^d$ be a graph where none of the rules \ref{rul:de1}, \ref{rul:de2}, \ref{rul:dt1}, \ref{rul:dt2} or \ref{rul:dt3} can be applied. Then $G$ contains no vertices of degree 1.
\end{lemma}
\begin{proof}
Assume on the contrary that $G$ does contain a vertex $v$ with degree 1. By Lemma \ref{ConflictingLabelsLemma} and the fact that $G\in \mathcal{G}_2^d$ we have that at most one constraint contains $v$. If there is no constraint on $v$, we could apply Rule \ref{rul:dt1} to it, therefore there is exactly one constraint on $v$. If this constraint is $(v,v)$ we would be able to apply Rule \ref{rul:dt2} to $v$. If the constraint is $(v,w)$, where $w$ is another vertex of degree 1, Rule \ref{rul:dt3} could be applied to $v$. The only remaining possibility is that the constraint on $v$ is the constraint $(v,w)$ where $w$ is a vertex with degree greater than 1. We will use $D$ to denote the divisor with $D(v)=D(w)=1$. Since we have the constraint $(v,w)$ and $G\in \mathcal{G}_2^d$, $D$ is a suitable divisor.

We first consider the case where $w$ is not a cut-vertex. Let $u$ be the neighbor of $v$. Consider the transformation from $D$ to a divisor $D'$ with $D'(u)=1$. Let $A_1$ be the first firing set in the level decomposition of this transformation. Note that we have $v,w\in A_1$ and $u\notin A_1$. Since $w$ is not a cut-vertex, it follows for each neighbor $w_i$ of $w$ that either there is a path from $w_i$ to $u$ that does not
contain $w$ or $w_i=u$. Note that if a neighbor $w_i\neq u$ is in $A_1$, then somewhere on its path to $u$ must be an edge that crosses between $A_1$ and its complement $A_1^c$. But such a crossing edge would imply that the firing set is not valid, since no vertex on this path contains a chip. Since $w$ has degree at least two, and none of its neighbors are in $A_1$, it follows that the firing set is not valid, since $w$ would lose at least two chips. We have a contradiction.

We proceed with the case where $w$ is a cut-vertex. Let $C_x$ be a connected component not containing $v$ after removing $w$. Consider the subset $C_x$ in $G$. Note that from $D$ we can never obtain an equivalent divisor with two chips on $C_x$. Since the chip from $v$ would have to move through $w$ to get to $C_x$, this would require $D$ to be equivalent to a divisor with two chips on $w$, which is impossible by Lemma \ref{ConflictingLabelsLemma} if $G\in \mathcal{G}_2^d$. Since $D$ reaches all vertices, it follows that $C_x$ must be a tree. This means $C_x$ must contain a vertex $x$ of degree one, we know however that since we cannot apply rules \ref{rul:dt1}, \ref{rul:dt2} or \ref{rul:dt3} to $G$, $x$ must have a constraint $(x,y)$ where $y$ is a vertex with degree greater than one. We now consider the possible locations of $y$.

If $y\in C_x$, then $D$ must be equivalent to a divisor with a chip on $x$ and a chip on $y$. As mentioned before, $D$ cannot be equivalent to a divisor with two chips on $C_x$, so it follows that $y\notin C_x$.

Since $y\notin C_x$, $D$ has to be equivalent to the divisor $D''$ with $D''(x)=D''(y)=1$. Let $C_y$ be the component containing $y$. Let $A_1$ be the first subset of the level set decomposition of the transformation of $D$ into $D''$. Note that $v,w\in A_1$ and $x,y\notin A_1$. But this implies that $w$ has at least one neighbor $w_1$ in $C_y$, with $w_1\notin A_1$, namely the first vertex on the path from $w$ to $y$. But $w$ also has at least one neighbor $w_2$ in $C_x$, with $w_2\notin A_1$, namely the first vertex on the path from $w$ to $x$. This means $w$ has two neighbors that it will send a chip to, but $w$ only has one chip. This yields a contradiction. 

We conclude that no vertices with degree $1$ can exist in $G$.
\end{proof}

\begin{lemma} \label{cycleLemma}
Let $G$ be a graph with a set of constraints $\mathscr{C}$ and let $C$ be a cycle in $G$ with $\mathscr{C}_C$ the set of constraints that contain a vertex in $C$. If $G\in \mathcal{G}_2^d$ then the constraints $\mathscr{C}_C$ are compatible.
\end{lemma}
\begin{proof}
We start by showing that the first property of a compatible constraint set holds. Let $(v,w)\in \mathscr{C}_C$ be a constraint and let $v\in C$ without loss of generality. We show that $w\in C$. Assume on the contrary that $w\notin C$, then let $D$ be the suitable divisor with $D(v)=D(w)=1$. Let $x$ be a vertex in $C$ with $x\neq v$. Let $D'\sim_\mathscr{C} D$ be a divisor with $D'(x)\geq 1$. Let $A_1$ be the first firing set of the level set decomposition of the transformation of $D$ into $D'$. Note that $v,w\in A_1$ and $x\notin A_1$. Note there are two disjoint paths from $v$ to $x$, since they are on the same cycle. This implies a chip will be sent along both these paths by $A_1$, but since $w\notin C$, both these chips must come from $v$. However, $v$ only has one chip, a contradiction. We conclude that $w\in C$.

For the second property, let $(v,w), (v', w')\in \mathscr{C}_C$ be two constraints on $C$. By our first property we have that $v,w,v',w'\in C$. Let $D$ be the divisor with $D(v)=D(w)=1$ and $D'$ the divisor with $D'(v')=D'(w')=1$. We know $D$ and $D'$ must be equivalent, since $G\in \mathcal{G}_2^d$ and both correspond to constraints on $G$. Let $A_1, \dots, A_k$ be the level set decomposition of the transformation of $D$ into $D'$. Note that $v,w\in A_1$ and $v', w'\notin A_1$. We observe that $v$ and $w$ split $C$ into two arcs. Note that both $v'$ and $w'$ must be on the same arc: if they are not on the same arc, there exists disjoint paths from $v$ to $v'$ and to $w'$ that do not contain $w$. This implies that $A_1$ sends two chips along these paths, but $v$ has only one chip. 

Now note that $C$ is biconnected, which implies that for a firing set $A$ with $w\in A$ and $w'\notin A$ to be valid there must be at least two chips on vertices in $C$. This follows since there are at least two edges crossing between $A$ and its complement $A^c$ in $C$. Since each of the firing sets $A_1, \ldots, A_k$ is valid, it follows this transformation leaves two chips on $C$ at each intermediate divisor. It follows that if we restrain these firing sets to $C$, we have a sequence of firing sets that transforms $D$ into $D'$ on $C$. Therefore $D$ and $D'$ are equivalent on $C$, so our second property is also fulfilled.
\end{proof}

\begin{lemma}[{Folklore, see e.g., \cite[Lemma 4]{BodlaenderK2011}}] \label{degreeLemma}
Let $G$ be a simple graph of treewidth 2 or lower and containing at least 4 vertices, then $G$ has at least two vertices with degree 2 or lower.
\end{lemma}

\begin{lemma} \label{ruleLemma}
Given a non-empty graph $G\in \mathcal{G}_2^d$ there is a rule in $\mathcal{R}^d$ that can be applied to $G$.
\end{lemma}
\begin{proof}
Let $G\in \mathcal{G}_2^d$ be such a graph and assume that no rule in $\mathcal{R}^d$ can be applied to $G$. By Lemma \ref{DegreeOneLemma} it follows that all vertices of $G$ have degree at least $2$. 
Consider the minor $H$ of $G$ created by contracting each path of degree $2$ vertices to an edge. Then any edge in $H$ was either created by contraction of a path of any number of vertices with degree $2$ in $G$ or it already was an edge in $G$. 

If $H$ contains a loop, there is a path of degree $2$ vertices in $G$ going from a degree $3$ or greater vertex to itself (since $G$ contains no loops). So this path plus the vertex it is attached to forms a cycle with exactly one vertex of degree $3$ or greater. Since we cannot apply Rule $\bm{C_2}$ to $G$, it follows that the constraints $\mathscr{C}_C$ are not compatible. This contradicts Lemma \ref{cycleLemma}.
Hence $H$ contains no loops.
 
Now we find a subgraph $H'$ of $H$ with no multiple edges. 
If $H$ contains no multiple edges, simply let $H'=H$. Otherwise let $v$ and $w$ be two vertices such that there are at least two edges between $v$ and $w$. Suppose that $v$ and $w$ are still connected to each other after removing two edges $e_1, e_2$ between them. The removed edges each represent a single edge or a path of degree $2$ vertices in $G$. Thus ${v,w}$ plus these paths form a cycle $C$ in $G$ with exactly two vertices of degree $3$ or greater, where there is also a path between $v$ and $w$ that does not share any edges with $C$. Since we cannot apply Rule $\bm{C_3}$ to $G$, it follows that the constraints $\mathscr{C}_C$ are not compatible. Again, this contradicts Lemma \ref{cycleLemma}. It follows that $G$ must be disconnected after removing $e_1$ and $e_2$. So any multiple edge in $H$ is a double edge, whose removal splits the graph in two connected components. Let $H'$ be the connected component of minimal size over all possible removals of a double edge in $H$. Note that $H'$ cannot contain any double edge, since this would imply a smaller connected component.

We now have a minor $H'$ of $G$, which is a simple graph since it has no loops or multiple edges. Also, each vertex of $H'$ has degree at least $3$ with at most one exception, namely the vertex that was incident to the two parallel edges that were removed to obtain $H'$.  Since a graph with treewidth at most two has at least two vertices of degree at most two, it follows by Lemma \ref{degreeLemma} that $\tw(H')\geq 3$. Since treewidth is closed under taking minors we get $\tw(G) \geq 3$. But then, since treewidth is a lower bound \cite{Gijs}, it follows that $\dgon(G) \geq 3$, creating a contradiction, since $G\in \mathcal{G}_2^d$. We conclude that our assumption must be wrong and there must be a rule in $\mathcal{R}^d$ that can be applied to $G$.
\end{proof}

Now we have everything required to prove our main theorem:

\begin{proof}[Proof of Theorem \ref{mainTheorem}]
By Lemma \ref{lem:divSafe} we have that $\mathcal{R}^d$ is safe. It remains to prove that $\mathcal{R}^d$ is also complete.

Assume that $G\in \mathcal{G}_2^d$. By Lemma \ref{ruleLemma} and Lemma \ref{lem:divSafe} we can keep applying rules from $\mathcal{R}^d$ to $G$ as long as $G$ has not been turned into the empty graph yet. Observe that each rule removes at least one vertex or at least two edges, while never adding more vertices or edges. Since $G$ is finite, rules from $\mathcal{R}^d$ can only be applied a finite number of times. When no more rules can be applied, it follows that the graph has been reduced to the empty graph. Therefore $\mathcal{R}^d$ is complete. 
\end{proof}

Hence $\mathcal{R}^d$ has the properties we want it to have so that we are able to use it for characterization of the graphs with divisorial gonality two or lower.

\section{Reduction Rules for Stable Gonality}
\label{section:stablegonalityrules}

In this section, we give a complete set of safe reduction rules to recognize stable hyperelliptic graphs, i.e.\ graphs with stable gonality 2. We will first introduce the constraints for stable gonality and then we will state all rules. Next we will show that all rules are safe for graphs with stable gonality at most 2 and that those graphs can be reduced to the empty graph. It is not hard to see that the set of rules implies a polynomial time algorithm to test if a graph has stable gonality at most 2; in Section~\ref{section:algorithms}, we discuss how we can obtain an algorithm with a running time of $O(m+n \log n)$.

\subsection*{Constraints for Stable Gonality}

For a given graph $G$, we want to know whether there exists a finite harmonic morphism of degree 2 from a refinement of $G$ to a tree. We will do this by reducing $G$ to the empty graph. During this process we sometimes add constraints to our graph. The set of constraints gives restrictions to which morphisms we allow. 
\begin{definition}
Let $G$ be a graph, $G'$ a refinement of $G$, $T$ a tree. Let $\phi\colon G' \to T$ be a map. We call $\phi$ a \emph{suitable morphism} if it is a finite harmonic morphism of degree 2 and it satisfies the following conditions.
\begin{enumerate} 
\item[\textup{(i)}] For all pairs $(v,v) \in \mathscr{C}$ it holds that $m_\phi(v) = 2$.
\item[\textup{(ii)}] For all pairs $(u,v) \in \mathscr{C}$ with $u \neq v$ it holds that $\phi(u) = \phi(v)$ and $m_\phi(u) = m_\phi(v) = 1$.
\end{enumerate}
\end{definition}

We say that a graph with constraints has stable gonality at most 2 if there exists a suitable morphism from a refinement of $G$ to a tree. Let $\mathcal{G}_2^s$ be the class of graphs with constraints that have stable gonality at most 2. We define the empty graph to have stable gonality 0 and thus $\emptyset \in \mathcal{G}_2^s$.

We will denote the set of constraints that contain a vertex $v$ by $\mathscr{C}_v$.

\subsection*{Reduction rules}
\label{subsec:rules}
	
We will now state all rules. Figure \ref{fig:rules} shows all rules in pictures, constraints are showed as green dashed edges. Sometimes it is convenient to think of constraints as an extra set of edges, from now on we will refer to the constraints as green edges.
We apply those rules to a given graph $G$ with an empty set of constraints. When a rule adds a constraint $uv$, and there already exists such a constraint, then the set of constraints does not change.

\begin{figure}[t]
		\centering
		\begin{tabular}{|rr|}
        \hline
			\multicolumn{1}{|l|}{\small{Rule \ref{rul:leaf}}} & \multicolumn{1}{l|}{\small{Rule \ref{rul:leafGreenLoop}}} \\
			\multicolumn{1}{|r|}{\begin{tikzpicture}
	\cirkelsEnPijl
	\node[vertexx] (u) at (1.25,0) {};
	\node[vertexx] (v) at (0,0) {};
	\draw[edge] (u)--(v);
	\node[vertexx] (u') at (4.25,0) {};
\end{tikzpicture} 
			} & \multicolumn{1}{r|}{\begin{tikzpicture}
	\cirkelsEnPijl
	\lus{groen}{0}{0}
	\lus{groen}{4.25}{0}
	\node[vertexx] (u) at (1.25,0) {};
	\node[vertexx] (v) at (0,0) {};
	\node[vertexx] (u') at (4.25,0) {};
	\draw[edge] (u)--(v);
\end{tikzpicture} } \\
			\hline
            \multicolumn{1}{|l}{\small{Rule \ref{rul:degree2}}} & \\
			\multicolumn{1}{|r}{\begin{tikzpicture}
	\cirkelsEnPijl
	\lus{edge}{4.25}{0}
	\node[vertexx] (w) at (1.25,0) {};
	\node[vertexx] (v) at (0,0) {};
	\draw[edge] (v) to [relative, out=40, in=140] (w);
	\draw[edge] (w) to [relative, out=40, in=140] (v);
	\node[vertexx] (u') at (4.25,0) {};
\end{tikzpicture} } & \multicolumn{1}{r|}{\begin{tikzpicture}
	\cirkelsEnPijl
	\node[vertexx] (w1) at (1.25,.375) {};
	\node[vertexx] (w2) at (1.25,-.375) {};
	\node[vertexx] (v) at (0,0) {};
	\draw[edge] (v)--(w1) (w2) -- (v);
	\node[vertexx] (u') at (4.25,0.375) {};
	\node[vertexx] (v') at (4.25,-.375) {};
	\draw[edge] (u') -- (v');
\end{tikzpicture} } \\
			\hline
            \multicolumn{1}{|l}{\small{Rule \ref{rul:twoleaves}}} & \\
			\multicolumn{1}{|r}{\begin{tikzpicture}
	\cirkelsEnPijl
	\lus{groen}{4.25}{0}
	\node[vertexx] (w) at (1.25,0) {};
	\node[vertexx] (u) at (0,.375) {};
	\node[vertexx] (v) at (0,-.375) {};
	\draw[edge] (u)--(w) -- (v);
	\draw[groen] (u) --(v);
	\node[vertexx] (u') at (4.25,0) {};
\end{tikzpicture} } & \multicolumn{1}{r|}{\begin{tikzpicture}
	\cirkelsEnPijl
	\node[vertexx] (w1) at (1.25,.375) {};
	\node[vertexx] (w2) at (1.25,-.375) {};
	\node[vertexx] (u) at (0,.375) {};
	\node[vertexx] (v) at (0,-.375) {};
	\draw[edge] (u)--(w1) (w2) -- (v);
	\draw[groen] (u) --(v);
	\node[vertexx] (u') at (4.25,0.375) {};
	\node[vertexx] (v') at (4.25,-.375) {};
	\draw[groen] (u') -- (v');
\end{tikzpicture} } \\
			\hline
            \multicolumn{1}{|l}{\small{Rule \ref{rul:degree2GreenLoop}}} & \\
			\multicolumn{1}{|r}{\begin{tikzpicture}
	\cirkelsEnPijl
	\lus{groen}{4.25}{0}
	\lus{groen}{0}{0}
	\node[vertexx] (w) at (1.25,0) {};
	\node[vertexx] (v) at (0,0) {};
	\draw[edge] (v) to [relative, out=40, in=140] (w);
	\draw[edge] (w) to [relative, out=40, in=140] (v);
	\node[vertexx] (u') at (4.25,0) {};
\end{tikzpicture} } & \multicolumn{1}{r|}{\begin{tikzpicture}
	\cirkelsEnPijl
	\lus{groen}{0}{0}
	\node[vertexx] (w1) at (1.25,.375) {};
	\node[vertexx] (w2) at (1.25,-.375) {};
	\node[vertexx] (v) at (0,0) {};
	\draw[edge] (v)--(w1) (w2) -- (v);
	\node[vertexx] (u') at (4.25,0.375) {};
	\node[vertexx] (v') at (4.25,-.375) {};
	\draw[groen] (u') -- (v');
	\draw[dotted, rounded corners=10] (w2) -- (2,-.5) -- (1.75,.5)  -- (w1);
	\draw[dotted, rounded corners=10] (v') -- (5,-.5) -- (4.75,.5)  -- (u');
\end{tikzpicture} }  \\
			\hline
            \multicolumn{1}{|l|}{\small{Rule \ref{rul:loops}}} & \multicolumn{1}{l|}{\small{Rule \ref{rul:parallelToGreenEdge}}} \\
			\multicolumn{1}{|r|}{\begin{tikzpicture}
	\cirkelsEnPijl
	\lus{edge}{1.25}{0}
	\lus{groen}{4.25}{0}
	\node[vertexx] (u) at (1.25,0) {};
	\node[vertexx] (u') at (4.25,0) {};
\end{tikzpicture} } & \multicolumn{1}{r|}{\begin{tikzpicture}
	\cirkelsEnPijl
	\node[vertexx] (u) at (1.25,.375) {};
	\node[vertexx] (v) at (1.25,-.375) {};
	\node[vertexx] (u') at (4.25,0.375) {};
	\node[vertexx] (v') at (4.25,-.375) {};
	\draw[groen] (u) to [relative, out=40, in=140] (v);
	\draw[edge] (v) to [relative, out=40, in=140] (u);
	\draw[groen] (u') -- (v');
\end{tikzpicture} } \\
			\hline
            \multicolumn{1}{|l|}{\small{Rule \ref{rul:twoEdgesAndPath}}} & \multicolumn{1}{l|}{\small{Rule \ref{rul:sgonEnd1}}} \\
			\multicolumn{1}{|r|}{\begin{tikzpicture}
	\cirkelsEnPijl
	\node[vertexx] (u) at (1.25,.375) {};
	\node[vertexx] (v) at (1.25,-.375) {};
	\node[vertexx] (u') at (4.25,0.375) {};
	\node[vertexx] (v') at (4.25,-.375) {};
	\draw[edge] (u) to [relative, out=40, in=140] (v);
	\draw[edge] (v) to [relative, out=40, in=140] (u);
	\draw[dotted, rounded corners=10] (v) -- (2,-.5) -- (1.75,.5)  -- (u);
	\draw[groen] (u') -- (v');
	\draw[dotted, rounded corners=10] (v') -- (5,-.5) -- (4.75,.5)  -- (u');
\end{tikzpicture} } & \multicolumn{1}{r|}{\begin{tikzpicture}
	\alleenPijl
	\node[vertexx] (u) at (1.25,0) {};
    \draw (4.25,0) circle (.2);
    \draw (4.45,0.2) -- (4.05, -0.2);
\end{tikzpicture} 
			} \\
            \hline 
            \multicolumn{1}{|l|}{\small{Rule \ref{rul:sgonEnd2}}} & \multicolumn{1}{l|}{\small{Rule \ref{rul:sgonEnd3}}} \\
			\multicolumn{1}{|r|}{\begin{tikzpicture}
	\alleenPijlKlein
    \lus{groen}{1.25}{0}
	\node[vertexx] (u) at (1.25,0) {};
	\draw (4.25,0) circle (.2);
    \draw (4.45,0.2) -- (4.05, -0.2);
\end{tikzpicture} 
			} & \multicolumn{1}{r|}{\begin{tikzpicture}
	\alleenPijlKlein
    \node[vertexx] (v) at (0,0) {};
	\node[vertexx] (u) at (1.25,0) {};
    \draw[groen] (u) -- (v);
	\draw (4.25,0) circle (.2);
    \draw (4.45,0.2) -- (4.05, -0.2);
\end{tikzpicture} 
			} \\
            \hline
        \end{tabular}
		\caption{The reduction rules for recognizing stable hyperelliptic graphs. }  
		\label{fig:rules}
	\end{figure}

\begin{customrule}{$\bm{T_1^s}$}\label{rul:leaf}
Let $v$ be a leaf with $\mathscr{C}_v = \emptyset$. Let $u$ be the neighbor of $v$. Contract the edge $uv$. 
\end{customrule}

\begin{customrule}{$\bm{T_2^s}$} \label{rul:leafGreenLoop}
Let $v$ be a leaf with $\mathscr{C}_v = \{(v,v)\}$. Let $u$ be the neighbor of $v$. Contract the edge $uv$. 
\end{customrule}

\begin{customrule}{$\bm{S_1^s}$} \label{rul:degree2}
Let $v$ be a vertex of degree 2 with $\mathscr{C}_v = \emptyset$. Let $u_1$, $u_2$ be the neighbors of $v$ (possibly $u_1 = u_2$). Contract the edge $u_1v$. 
\end{customrule}

\begin{customrule}{$\bm{T_3}^s$}\label{rul:twoleaves}
Let $G$ be a graph where every leaf and every degree 2 vertex is incident to a green edge. Let $v_1$ and $v_2$ be two leaves that are connected by a green edge. Let $u_1$ and $u_2$ be their neighbors (possibly $u_1 = u_2$). Contract the edges $u_1v_1$ and $u_2v_2$. 
\end{customrule}

\begin{customrule}{$\bm{S_2^s}$} \label{rul:degree2GreenLoop}
Let $G$ be a graph where every leaf and every degree 2 vertex is incident to a green edge. Let $v$ be a vertex of degree 2 with a green loop, such that there exists a path from $v$ to $v$ in $G$ (possibly containing green edges). Let $u_1$ and $u_2$ be the neighbors of $v$ (possibly $u_1 = u_2$). Remove $v$ and connect $u_1$ and $u_2$ with a green edge. 
\end{customrule}

\begin{customrule}{$\bm{L^s}$} \label{rul:loops}
Let $v$ be a vertex with a loop. Remove the loop from $v$ and add a green loop to $v$.  
\end{customrule}

\begin{customrule}{$\bm{P_1^s}$} \label{rul:parallelToGreenEdge}
Let $uv$ be an edge such that there also exists a green edge $uv$. Remove the black edge $uv$.   
\end{customrule}

\begin{customrule}{$\bm{P_2^s}$} \label{rul:twoEdgesAndPath}
Let $u,v$ be vertices, such that $|E(u,v)|>1$. Let $e$ and $f$ be two of those edges. If there exists another path, possibly containing green edges, from $u$ to $v$, then remove $e$ and $f$ and add a green edge from $u$ to $v$. 
\end{customrule}

\begin{customrule}{$\bm{E_1^s}$} \label{rul:sgonEnd1}
Let $G$ be the graph consisting of a single vertex $v$ with $\mathscr{C}_v = \emptyset$. Remove $v$. 
\end{customrule}

\begin{customrule}{$\bm{E_2^s}$} \label{rul:sgonEnd2}
Let $G$ be the graph consisting of a single vertex $v$ with a green loop. Remove $v$. 
\end{customrule}

\begin{customrule}{$\bm{E_3^s}$} \label{rul:sgonEnd3}
Let $G$ be the graph consisting of a two vertices $u$ and $v$ that are connected by a green edge. Remove $u$ and $v$. 
\end{customrule}

We will write $\mathcal{R}^s$ for this set of reduction rules. We can now state the main theorem; in the next sections we will prove this theorem. 

\begin{theorem} \label{thm:sgonSafeComplete}
The set of rules $\mathcal{R}^s$ is safe and complete for $\mathcal{G}_2^s$. 
\end{theorem}

\subsection*{Safeness}
    
Now we will prove that the rules $\mathcal{R}^s$ are safe for $\mathcal{G}_2^s$, i.e., if $G$ a is graph, and $H$ is obtained from $G$ by applying one of the rules, then $\sgon(G) \leq 2$ if and only if $\sgon(H) \leq 2$. In all proofs we assume that the original graph is called $G$ and the graph obtained by applying a rule is called $H$.

\begin{lemma} \label{thm:firstRule}
Rule \ref{rul:leaf} is safe. 
\end{lemma}

\begin{proof}
Let $v$ be the leaf in $G$ to which the rule is applied. 

Suppose that $\sgon(G) \leq 2$. Since $G$ is a refinement of $H$, it is clear that $\sgon(H) \leq 2$. 
		
Suppose that $\sgon(H) \leq 2$. Then there exists a refinement $H'$ of $H$ and a suitable morphism $\phi\colon H' \to T$. Write $u$ for the neighbor of $v$ in $G$. We distinguish two cases. 
		
Suppose that $m_\phi(u) = 2$. Then add a leaf $v$ to $u$ in $H'$ to obtain $G'$. Now we see that $G'$ is a refinement of $G$. Give the edge $uv$ index $r_{\phi'}(uv) = 2$, and give all other edges $e$ index $r_{\phi'}(e) = r_\phi(e)$. Add a leaf $v'$ to $\phi(u)$ in $T$ to obtain $T'$. Then we can extend $\phi$ to $\phi'\colon G' \to T'$,
\begin{align*}
	\phi'(x) = \begin{cases}
	\phi(x) & \text{if } x \in H'\\
	v'& \text{if } x = v. 
	\end{cases}
\end{align*} 
It is clear that $\phi'$ is a suitable morphism, so we conclude that $\sgon(G) \leq 2$. 
		
Suppose that $m_\phi(u) = 1$. Let $w$ be the other vertex such that $\phi(w) = \phi(u)$. Then add leaves $v_1$ and $v_2$ to $u$ and $w$ in $H'$ to obtain $G'$. We see that $G'$ is a refinement of $G$. Give the edges $uv_1$ and $wv_2$ indices $r_{\phi'}(uv_1) = r_{\phi'}(wv_2) = 1$, and give all other edges $e$ index $r_{\phi'}(e) = r_{\phi}(e)$. Add a leaf $v'$ to $\phi(u)$ in $T$ to obtain $T'$. Then we can extend $\phi$ to $\phi'\colon G' \to T'$,
\begin{align*}
	\phi'(x) = \begin{cases}
	\phi(x) & \text{if } x \in H'\\
	v'& \text{if } x = v_1, v_2. 
	\end{cases}
\end{align*} 
It is clear that $\phi'$ is a suitable morphism, so we conclude that $\sgon(G) \leq 2$. 
\end{proof}

\begin{lemma} \label{lem:leafGreenLoop}
Rule \ref{rul:leafGreenLoop} is safe. 
\end{lemma}

\begin{proof}
Let $v$ be the vertex in $G$ to which the rule is applied. 

Suppose that $\sgon(G) \leq 2$. Then there exists a refinement $G'$ of $G$ and a suitable morphism $\phi\colon G' \to T$. Let $u$ be the neighbor of $v$ in $G$. We distinguish two cases: 
		
Suppose that $m_\phi(u) = 2$. Define $H'$ as the graph $G'$ with a green loop at vertex $u$ and without the green loop at $v$, then $H'$ is a refinement of $H$. Now we see that $\phi\colon H'\to T$ is a suitable morphism, so $\sgon(H) \leq 2$.  
		
\begin{figure}
		\centering
		\begin{tikzpicture}
		\draw[lightgray] (.6,0) circle [x radius=1, y radius=.4];
		\draw[lightgray] (.6,-1) circle [x radius=1, y radius=.4];
		\lus{groen}{-3}{0}
		\node[vertexx, label=$v$] (v) at (-3, 0) {};
		\node[vertexx, label=$u$] (u) at (0, 0) {};
		\node[addedd, label=$v_1$] (a) at (-2,0) {};
		\node[addedd, label=$v_2$] (b) at (-1,0) {};
		\node[addedd, label=below:$x_1$] (c) at (-1,-1) {};
		\node[addedd, label=below:$x$] (x) at (0,-1) {};
		\draw[edge] (u) -- (b) -- (a) -- (v);
		\draw[edge] (a) -- (c) --  (x);
		\node[vertex, label=$w$] (w) at (.4,.25) {};
		\node[vertex, label=below:$y$] (y) at (.4,-.75) {}; 
		\draw[edge] (u) -- (w);
		\draw[edge] (x) -- (y);
		\end{tikzpicture}
		\caption{Proof of Lemma \ref{lem:leafGreenLoop}.} \label{fig:bewijsRule2}
	\end{figure}
		
Suppose that $m_\phi(u) = 1$. Let $v_0 = v, v_1, \ldots, v_k=u$ be the vertices that are added to the edge $uv$ of $G$. Let $i$ be the largest integer such that $m_\phi(v_i) = 2$. Notice that $i<k$. Then there exists another vertex $x_1$ in $G'$ such that $\phi(v_{i+1}) = \phi(x_1)$. If $v_{i+1} \neq u$, it follows that there is an edge $x_1x_2$ that is mapped to $\phi(v_{i+1}v_{i+2})$. And since $m_\phi(v_{i+2}) = 1$, we see that $x_2 \neq v_{i+2}$. It follows that there exists $x_1\neq v_{i+1}$, $\ldots$, $x_{k-i} \neq v_{k}$ such that $\phi(v_{i+j}) = \phi(x_j)$. Write $x = x_{k-i}$, then $\phi(x) = \phi(u)$ and $m_\phi(u) = m_\phi(x) = 1$. See Figure \ref{fig:bewijsRule2} for an illustration of this. 
		
Notice that $x$ is an external added vertex. Let $w$ be a neighbor of $u$ not equal to $v_{k-1}$. Then we see that there exists an vertex $y$ such that $\phi(uw) = \phi(xy)$. Since $x$ is an external added vertex, we see that $w\neq y$. We conclude that $m_\phi(w) = 1$. Inductively we see that for every vertex $w'$ in $G_{v_i}(v_{i+1})\backslash\{v_i\}$ it holds that $m_\phi(w') = 1$. Define $H'$ as $G_{v_i}(v_{i+1})\backslash\{v_i\}$, with a green loop at vertex $u$. Notice that $H'$ is a refinement of $H$. Now we can restrict $\phi$ to $H'$ and give every edge index $r_{\phi'}(e) = 2$ to obtain a suitable morphism: $\phi'\colon H' \to T'$, where $T' = \phi(G_{v_i}(v_{i+1})\backslash\{v_i\})$. We conclude that $\sgon(H) \leq 2$. 
		
Suppose that $\sgon(H) \leq 2$. Then there exists a refinement $H'$ of $H$ and a suitable morphism $\phi\colon H' \to T$. Write $u$ for the neighbor of $v$ in $G$. We know that $m_\phi(u) = 2$. Then add a leaf with a green loop to $u$ in $H'$ to obtain $G'$. Now we see that $G'$ is a refinement of $G$. Give the edge $uv$ index $r_{\phi'}(uv) = 2$, and give all other edges $e$ index $r_{\phi'}(e) = r_{\phi}(e)$. Add a leaf $v'$ to $\phi(u)$ in $T$ to obtain $T'$. Then we can extend $\phi$ to $\phi'\colon G' \to T'$,
\begin{align*}
	\phi'(x) = \begin{cases}
	\phi(x) & \text{if } x \in H'\\
	v'& \text{if } x = v. 
	\end{cases}
\end{align*} 
It is clear that $\phi'$ is a suitable morphism, so we conclude that $\sgon(G) \leq 2$. 
\end{proof}

\begin{lemma}
Rule \ref{rul:degree2} is safe. 
\end{lemma}

\begin{proof}
Let $v$ be the vertex in $G$ to which the rule is applied. 

Suppose that $\sgon(G) \leq 2$. Let $G'$ be a refinement of $G$ such that there exists a suitable morphism $\phi\colon G'\to T$. Since $G'$ is a refinement of $H$ too, it is clear that $\sgon(H) \leq 2$. 
		
Now suppose that $\sgon(H) \leq 2$. Let $H'$ be a refinement of $H$ such that there exists a suitable morphism $\phi\colon H'\to T$. Write $u_1$ and $u_2$ for the neighbors of $v$ in $G$. We distinguish two cases. 
Suppose that $u_1 = u_2$. It follows that the edge $u_1u_2$ is subdivided in $H'$, thus $H'$ is a refinement of $G'$ too. We conclude that $\sgon(G) \leq 2$. 
Suppose that $u_1 \neq u_2$. If the edge $u_1u_2$ is subdivided, we see again that $H'$ is a refinement of $G$, and $\sgon(G) \leq 2$. So suppose that $u_1$ and $u_2$ are neighbors in $H'$. Again, we distinguish two cases. 

If $r(u_1u_2) = 2$, then add a vertex $v$ on the edge $u_1u_2$ to obtain a graph $G'$. Notice that $G'$ is a refinement of $G$. Give the edges $u_1v$ and $vu_2$ index 2, and give all other edges $e$ index $r_{\phi'}(e) = r_{\phi}(e)$. And add a vertex $v'$ on the edge $\phi(u_1)\phi(u_2)$ in $T$ to obtain $T'$. Now we see that $\phi'\colon G' \to T'$ given by 
\begin{align*}
	\phi'(x) = \begin{cases}
	\phi(x) & \text{if } x \in H',\\
	v'& \text{if } x = v, 
	\end{cases}
\end{align*} 
is a suitable morphism. We conclude that $\sgon(G) \leq 2$. 

If $r(u_1u_2) = 1$, then there exists another edge $w_1w_2$ such that $\phi(u_1u_2) = \phi(w_1w_2)$. Now add a vertex $v_1$ on the edge $u_1u_2$  and a vertex $v_2$ on the edge $w_1w_2$ to obtain a graph $G'$. Notice that $G'$ is a refinement of $G$. Give the edges $u_1v_1$, $v_1u_2$, $w_1v_2$ and $v_2w_2$ index 1, and give all other edges $e$ index $r_{\phi'}(e) = r_{\phi}(e)$. Add a vertex $v'$ on the edge $\phi(u_1)\phi(u_2)$ in $T$ to obtain $T'$. Now we see that $\phi'\colon G' \to T'$ given by 
\begin{align*}
	\phi'(x) = \begin{cases}
	\phi(x) & \text{if } x \in H',\\
	v'& \text{if } x = v_1, v_2, 
	\end{cases}
\end{align*}  is a suitable morphism. We conclude that $\sgon(G) \leq 2$.
\end{proof}

\begin{lemma}\label{lem:twoleaves}
Rule \ref{rul:twoleaves} is safe.
\end{lemma} 

\begin{proof}
Let $v_1$ and $v_2$ be the vertices in $G$ to which the rule is applied. 

\begin{figure}
	\centering
	\begin{subfigure}{0.3 \textwidth}
		\begin{tikzpicture}
		\draw[lightgray] (3.7,0) circle [x radius=1, y radius=.3];
		\draw[lightgray] (3.7,-.75) circle [x radius=1, y radius=.3];
		\node[vertexx, label=$v_1$] (v1) at (0, 0) {};
		\node[addedd, label=$a_1$] (a1) at (1,0) {};
		\node[addedd, label=$a_2$] (a2) at (2,0) {};
		\node[vertexx, label=$u_1$] (u1) at (3, 0) {};
		\node[vertexx, label=below:$v_2$] (v2) at (0, -.75) {};
		\node[addedd, label=below:$b_1$] (b1) at (.75,-.75) {};
		\node[addedd, label=below:$b_2$] (b2) at (1.5,-.75) {};
		\node[addedd, label=below:$b_3$] (b3) at (2.25, -.75) {};
		\node[vertexx, label=below:$u_2$] (u2) at (3, -.75) {};
		\draw[edge] (v1) -- (a1) -- (a2) -- (u1);
		\draw[edge] (v2) -- (b1) -- (b2) -- (b3) -- (u2);
		\draw[groen] (v1) -- (v2);
		\node[addedd, label=below:$c$] (c) at (.4,-1.3) {}; 
		\draw[edge] (v2) -- (c);
		\end{tikzpicture} 
	\caption{Case 1.}
    \label{fig:bewijsRule4a}
    \end{subfigure} \hspace{0.1\textwidth}
    \begin{subfigure}{0.3\textwidth}
    	\begin{tikzpicture}
		\draw[lightgray] (3.25,-.375) circle (.75);
		\node[vertexx, label=$v_1$] (v1) at (0, 0) {};
		\node[addedd, label=$a_1$] (a1) at (1,0) {};
		\node[addedd, label=$a_2$] (a2) at (2,0) {};
		\node[vertexx, label=$u_1$] (u1) at (3, 0) {};
		\node[vertexx, label=below:$v_2$] (v2) at (0, -.75) {};
		\node[addedd, label=below:$b_1$] (b1) at (.75,-.75) {};
		\node[addedd, label=below:$b_2$] (b2) at (1.5,-.75) {};
		\node[addedd, label=below:$b_3$] (b3) at (2.25, -.75) {};
		\node[vertexx, label=below:$u_2$] (u2) at (3, -.75) {};
		\draw[edge] (v1) -- (a1) -- (a2) -- (u1);
		\draw[edge] (v2) -- (b1) -- (b2) -- (b3) -- (u2);
		\draw[groen] (v1) -- (v2);
		\node[vertexx, label=below:$x$] (x) at (3.6,-.2) {}; 
		\node[addedd, label=below:$w$] (w) at (2.65,-1.3) {}; 
		\draw[edge] (u1) -- (x);
		\draw[edge] (b3) -- (w);
		\end{tikzpicture}
	\caption{Case 2.}
    \label{fig:bewijsRule4b}
	\end{subfigure}
	\caption{Proof of Lemma \ref{lem:twoleaves}.} 
\end{figure}
``$\Longrightarrow$'': Suppose that $\sgon(G) \leq 2$. Let $G'$ be a minimum refinement of $G$ such that there exists a suitable morphism $\phi\colon G' \to T$, i.e.\ for every refinement $G''$ with less vertices than $G'$ there is no suitable morphism $\phi'\colon G'' \to T'$ for any tree $T'$.  Let $u_1$ and $u_2$ be the neighbors of $v_1$ and $v_2$ in $G$. We distinguish three cases. 
		
Case 1: Suppose that $u_1 \neq u_2$, and that there does not exist a path from $v_1$ to $v_2$ in black and green edges, except the green edge $v_1v_2$. Let $a_0 = v_1, a_1, \ldots, a_k = u_1$ be the subdivision of the edge $u_1v_1$ and $b_0 = v_2, b_1, \ldots, b_l = u_2$ the subdivision of the edge $u_2v_2$. (See Figure \ref{fig:bewijsRule4a}.) We know that there exists an edge $v_2c$ such that $\phi(a_0a_1) = \phi(v_2c)$. It is clear that $c \neq a_1$, thus $m_\phi(a_1) = 1$. Inductively we find that for every vertex $a'$ in $G'_{a_0}(a_1)$ it holds that $m_\phi(a') = 1$. We conclude that $G'_{a_0}(a_1)$ is a tree. Analogously we find that $G'_{b_0}(b_1)$ is a tree. Thus $G_{v_1}(u_1)$ and $G_{v_2}(u_2)$ are trees. Thus $H$ consists of two black trees connected by a green edge. 
		
Now we can construct a refinement $H'$ of $H$, a tree $T'$ and a suitable morphism $\phi'\colon H' \to T'$. Copy every branch of $u_1$ and add them to $u_2$ and copy every branch of $u_2$ and add them to $u_1$. Write $H'$ for this graph. Now we see that the two trees of $H'$ are the same, say $T'$. Now we can define $\phi'\colon H' \to T'$ as the identity map on each of the components, where $\phi'(u_1) = \phi'(u_2)$. Thus $\phi'$ is a suitable morphism. We conclude that $\sgon(H) \leq 2$. 
		
Case 2: Suppose that $u_1 \neq u_2$ and that there exists a path (possibly containing green edges) from $v_1$ to $v_2$. Assume that $\phi(u_1) \neq \phi(u_2)$. Let $a_0 = v_1, a_1, \ldots, v_k=u_1$ be the added vertices on the edge $v_1u_1$ and let $b_0 = v_2, b_1, \ldots, b_l=u_2$ be the added vertices on the edge $v_2u_2$. Assume without loss of generality that $k \leq l$. (See Figure \ref{fig:bewijsRule4b}.) It is clear that all vertices $a_0, \ldots, a_k, b_0, \ldots, b_l$ lie on the path from $v_1$ to $v_2$. 
Suppose that $\phi(a_1) \neq \phi(b_1)$. The path from $a_1$ to $b_1$ is mapped to a walk from $\phi(a_1)$ to $\phi(b_1)$. It is clear that $\phi(v_1)$ is contained in this walk, so there is a vertex $x$ in the path from $a_1$ to $b_1$ that is mapped to $\phi(v_1)$. This yields a contradiction. Thus $\phi(a_1) = \phi(b_1)$. Inductively we find that $\phi(a_i) = \phi(b_i)$ for all $i\leq k$. We conclude that $\phi(b_k) = \phi(u_1)$. Notice that $b_k \neq u_2$, since we assumed $\phi(u_1) \neq \phi(u_2)$. It follows by Lemma \ref{lem:morphism-degree} that $\deg(b_k) = \deg(u_1)$. We again distinguish two cases. 
		
Suppose that $\deg(u_1) > 2$. Then $b_k$ has an external added neighbor $w$. We see that $u_1$ has a neighbor $x$ such that $\phi(b_kw) = \phi(u_1x)$. Since $w$ is an external added vertex, it follows that $w \neq x$. Thus $m_\phi(w) = 1$. (See Figure \ref{fig:bewijsRule4b}.) Iteratively we see that for every vertex $w'$ in $G'_{b_k}(w)$, it holds that $m_\phi(w') = 1$. Notice that $G'_{b_k}(w)$ is a tree, since $w$ is an external added vertex. Now let $y$ be a leaf in $G'_{b_k}(w)$, and let $y'$ be such that $\phi(y) = \phi(y')$. Then $y'$ is a leaf. It is clear that $y'$ has no green edge incident to it, thus $y'$ is an added vertex. We conclude that we can remove $y$ and $y'$ from $G'$ and $\phi(y)$ from $T$ and still have a suitable morphism. This yield a contradiction with the minimality of $G'$. 
		
Suppose that $\deg(u_1) = 2$ in $G'$. Then the degree of $u_1$ in $G$ is also 2. It follows that $\mathscr{C}_{u_1} \neq \emptyset$. Let $u_1c$ be a green edge. If $c=u_1$, so if $u_1$ has a green loop, then $m_\phi(u_1) = 2$. This yields a contradiction. It is clear that $c\neq b_l$, since $b_l$ is an added vertex. It follows that there are 3 distinct vertices that are mapped to $\phi(u_1)$. This yields a contradiction. 
		
Altogether we conclude that $\phi(u_1) = \phi(u_2)$. Define $H'$ as $G'$ with a green edge $u_1u_2$. Now $H'$ is a refinement of $H$, and $\phi\colon H' \to T$ is a suitable morphism. We conclude that $\sgon(H) \leq 2$. 

Case 3: Suppose that $u_1 = u_2$. Analogous to the second case, we prove that $m_\phi(u_1) = 2$. Define $H'$ as $G'$ with a green loop at vertex $u_1$. Now $H'$ is a refinement of $H$, and $\phi\colon H' \to T$ is a suitable morphism. We conclude that $\sgon{H} \leq 2$. 

``$\Longleftarrow$'': 	
Suppose that $\sgon(H) \leq 2$. Then there exists a refinement $H'$ of $H$ and a suitable morphism $\phi\colon H' \to T$. Write $u_1$ and $u_2$ for the neighbors of $v_1$ and $v_2$ in $G$. We know that $\phi(u_1) = \phi(u_2)$. Then add a leaves $v_1$ and $v_2$ to $u_1$ and $u_2$ and a green edge $v_1v_2$ in $H'$ to obtain $G'$. Now we see that $G'$ is a refinement of $G$. Give the edges $u_1v_1$ and $u_2v_2$ index $r_{\phi'}(u_1v_1) = r_{\phi'}(u_2v_2) = 1$, and give all other edges $e$ index $r_{\phi'}(e) = r_{\phi}(e)$. Add a leaf $v'$ to $\phi(u_1)$ in $T$ to obtain $T'$. Then we can extend $\phi$ to $\phi'\colon G' \to T'$,
\begin{align*}
	\phi'(x) = \begin{cases}
	\phi(x) & \text{if } x \in H'\\
	v'& \text{if } x = v_1, v_2. 
	\end{cases}
\end{align*} 
It is clear that $\phi'$ is a suitable morphism, so we conclude that $\sgon(G) \leq 2$. 
\end{proof}

\begin{lemma}
Rule \ref{rul:degree2GreenLoop} is safe. 
\end{lemma}

\begin{proof}
This proof is analogous to the proof of the second and third case in the proof of Lemma \ref{lem:twoleaves}, so we omit it. 
\end{proof}

\begin{lemma}
Rule \ref{rul:loops} is safe. 
\end{lemma}

\begin{proof}
Let $v$ be the vertex in $G$ to which the rule is applied. 

Suppose that $\sgon(G) \leq 2$. Then there exists a refinement $G'$ of $G$ and a suitable morphism $\phi\colon G' \to T$. Let $u$ be a vertex that is added to the loop $vv$. We distinguish two cases. 
		
Suppose that $m_\phi(v) = 2$. Define $H'$ as the graph $G' \backslash G'_{v}(u)$ with a green loop at vertex $v$, then $H'$ is a refinement of $H$. Let $T' = T\backslash \phi(G'_v(u))$ we see that the restricted morphism $\phi\colon H'\to T'$ is a suitable morphism, so $\sgon(H) \leq 2$. 
		
Suppose that $m_\phi(u) = 1$. Let $v_0 = v, v_1, \ldots, v_k=v$ be the vertices that are added to the loop $vv$ of $G$. Let $i$ be the integer such that $\phi(v_i) = \phi(v)$. Let $w$ be a neighbor of $v$ not equal to $v_1$ or $v_{k-1}$. Then there is a neighbor $x$ of $v_i$, not equal to $v_{i-1}$ and $v_{i+1}$, such that $\phi(w) = \phi(x)$. Notice that $x$ is an external added vertex, thus $m_\phi(w) = m_\phi(x) = 1$. Inductively we see that for every vertex $w'$ in $G'_{v}(w)$ it holds that $m_\phi(w') = 1$. We conclude that $G'_{v}(w)$ is a tree.  

Define $H'$ as $G'_{v}(w)$, with a green loop at vertex $v$. Notice that $H'$ is a refinement of $H$. Now we can restrict $\phi$ to $H'$ and give every edge index $r_{\phi'}(e) = 2$ to obtain a suitable morphism: $\phi'\colon H' \to T'$, where $T' = \phi(G'_{v}(w))$. We conclude that $\sgon(H) \leq 2$. 
		
Suppose that $\sgon(H) \leq 2$. Then there exists a refinement $H'$ of $H$ and a suitable morphism $\phi\colon H' \to T$. We know that $m_\phi(v) = 2$. Then add a vertex $u$ to $H'$ with two black edges to $v$ to obtain $G'$. Now we see that $G'$ is a refinement of $G$. Give both edges $uv$ index $r_{\phi'}(uv) = 1$, and give all other edges $e$ index $r_{\phi'}(e) = r_{\phi}(e)$. Add a leaf $v'$ to $\phi(u)$ in $T$ to obtain $T'$. Then we can extend $\phi$ to $\phi'\colon G' \to T'$,
\begin{align*}
	\phi'(x) = \begin{cases}
	\phi(x) & \text{if } x \in H'\\
	v'& \text{if } x = u. 
	\end{cases}
\end{align*} 
It is clear that $\phi'$ is a suitable morphism, so we conclude that $\sgon(G) \leq 2$. 
\end{proof}

\begin{lemma} \label{thm:parallelToGreenEdge}
Rule \ref{rul:parallelToGreenEdge} is safe. 
\end{lemma}

\begin{proof}
Let $uv$ be the edge in $G$ to which the rule is applied. 

Suppose that $\sgon(G) \leq 2$. Let $G'$ be a refinement of $G$ and $\phi\colon G'\to T$ a suitable morphism. Let $G_{uv}$ be all internal and external added vertices to the edge $uv$. Now define $H' = G'\backslash G_{uv}$ and $T'= T\backslash(\phi(G_{uv}))$. Write $\phi'$ for the restriction of $\phi$ to $H'$. Notice that $\phi'$ is a suitable morphism and that $H'$ is a refinement of $H$. Thus $\sgon(H) \leq 2$. 
		
Suppose that $\sgon(H) \leq 2$. Let $H'$ be a refinement of $H$ and $\phi\colon H'\to T$ a suitable morphism. Add an edge $uv$ and a vertex $w$ on this edge to $H'$, to obtain a refinement $G'$ of $G$. Add a vertex $w'$ to $T$ with an edge to $\phi(u)$, to obtain tree $T'$. Give the edges $uw$ and $vw$ index $r_{\phi'}(uw) = r_{\phi'}(vw) = 1$, and give all other edges $e$ index $r_{\phi'}(e) = r_{\phi}(e)$. Look at $\phi'\colon G'\to T'$, defined as
\begin{align*}
	\phi'(x) = \begin{cases}
	\phi(x) & \text{if } x \in H'\\
	w' & \text{if } x = w.
	\end{cases}
\end{align*} Notice that $\phi(uw) = \phi(vw)$, since there is a green edge $uv$. We conclude that $\phi'$ is a suitable morphism, thus $\sgon(G) \leq 2$. 
\end{proof}

\begin{lemma} 
Rule \ref{rul:twoEdgesAndPath} is safe. 
\end{lemma}

\begin{proof}
Let $u$ and $v$ be the vertices in $G$ to which the rule is applied. 

Suppose that $\sgon(G) \leq 2$. Let $G'$ be a refinement of $G$ and $\phi\colon G'\to T$ a suitable morphism. If $\phi(u)\neq\phi(v)$, then there are at least three paths that are mapped to the path from $\phi(u)$ to $\phi(v)$ in $T$. This yields a contradiction. Thus $\phi(u) = \phi(v)$. Now we see, analogous to the proof of Lemma \ref{thm:parallelToGreenEdge}, that $\sgon(H) \leq 2$. 
		
Suppose that $\sgon(H) \leq 2$. Then we find analogous to the proof of Lemma \ref{thm:parallelToGreenEdge}, that $\sgon(G) \leq 2$. 
\end{proof}

\begin{lemma} \label{lem:End1}  \label{lem:End2}
 \label{lem:End3} \label{thm:lastRule}
Rules \ref{rul:sgonEnd1}, \ref{rul:sgonEnd2} and \ref{rul:sgonEnd3} are safe. 
\end{lemma}

\begin{proof}
All these graphs have stable gonality at most 2, so the statement holds true. 
\end{proof}

Now we have proven that all rules are safe, thus we have the following lemma:

\begin{lemma} \label{lem:sgonSafe}
The set of rules $\mathcal{R}^s$ is safe for $\mathcal{G}_2^s$. \qed 
\end{lemma}

\subsection*{Completeness}
	
Now we will prove that the set of rules is complete. 
If $G$ is a disconnected graph, then $\sgon(G) = 2$,  if and only if $G$ is a forest with two components. 
This can easily be checked. So we need to show that every connected graph can be reduced to the empty graph.  We will do this by showing that a rule can be applied to every connected graph with stable gonality 2 (Lemma \ref{lem:sgonComplete}). 
Since applying rules maintains connectivity (Lemma \ref{lem:sgon-connected}), this suffices. 

\begin{lemma} \label{lem:sgon-connected}
Let $G$ and $H$ be graphs. If $G$ is connected in the sense of black and green edges, and $H$ can be produced from $G$ by applying some rules, then $H$ is connected in the sense of black and green edges.
\end{lemma}
\begin{proof}
Notice that contracting an edge maintains connectivity, so for Rules \ref{rul:leaf}, \ref{rul:leafGreenLoop}, \ref{rul:degree2} and \ref{rul:twoleaves} this lemma holds. The lemma is also clearly true for Rules \ref{rul:sgonEnd1}, \ref{rul:sgonEnd2} and \ref{rul:sgonEnd3}. All the other rules introduce a green edge that makes sure that the graph $H$ is connected. 
\end{proof}	

First, we show two lemmas about constraints, that we need to show that we can apply a rule to every connected graph in $\mathcal{G}_2^s$. 

\begin{lemma} \label{prop:multipleGreenEdges}
Let $G$ be a graph with constraints. If there is a vertex $v$ with $|\mathscr{C}_v| > 1$, then $\sgon(G) \geq 3$. 
\end{lemma}

\begin{proof}
Let $G$ be a graph with $\sgon(G) = 2$. Suppose that $|\mathscr{C}_v| > 1$. Let $(u,v)$ and $(v, w)$ be two constraints that contain $v$. We know that $u \neq w$. Suppose that $\phi$ is a suitable morphism of degree 2. We distinguish two cases. 
Suppose that $u = v$. Then we know that $m_\phi(v) = 2$. On the other hand we have that $m_\phi(v) = m_\phi(w) = 1$. This yields a contradiction. 
Now suppose that $u\neq v$ and $w\neq v$. Notice that $\phi(u) = \phi(v) = \phi(w)$, thus there are at least three vertices mapped to $\phi(v)$. We conclude that $\deg(\phi) \geq 3$. This yields a contradiction. 
We conclude that $|\mathscr{C}_v| \leq 1$. 
\end{proof}

\begin{lemma}\label{prop:greenEdgeNotSameDegree} 
Let $G$ be a graph where every leaf is incident to a constraint, so if $\deg(u) = 1$ then $\mathscr{C}_u \neq \emptyset$ for all $u$. Suppose that $(u,v) \in \mathscr{C}$. If $\deg(u) \neq \deg(v)$, then $\sgon(G) \geq 3$. 
\end{lemma} 

\begin{proof}
Suppose that $\deg(u) \neq \deg(v)$. Assume without loss of generality that $\deg(u) > \deg(v)$. Suppose that $\sgon(G) = 2$. Let $G'$ be a refinement of $G$ with a minimal number of vertices such that there exists a suitable morphism of degree 2. Let $\phi\colon G' \to T$ be such a morphism. 
		
We know that $\phi(u) = \phi(v)$, thus, by Lemma \ref{lem:morphism-degree}, $\deg_{G'}(u) = \deg_{G'}(v)$. So there is a neighbor $x$ of $v$ which is an external added vertex. 
Now we look at $\phi(x)$. Notice that there is a neighbor $y$ of $u$ such that $\phi(x) = \phi(y)$. It is clear that $y\neq x$, since $x$ is an external added vertex. Thus $m_\phi(x) = m_\phi(y) = 1$. 
		
Let $x'$ be a neighbor of $x$, not equal to $v$. Suppose that $m_\phi(x') = 2$. We know that the edge $e= (x, x')$ has index $1$, so there there exists another neighbor of $x'$ that is mapped to $\phi(x)$. We know that $y$ is the unique vertex other than $x$ that is mapped to $\phi(x)$, it follows that $y$ is a neighbor of $x'$. This yields a contradiction, since $x'$ is an external added vertex. We conclude that $m_\phi(x') = 1$. Inductively we find that $m_\phi(x'') = 1$ for all vertices $x'' \in G_v(x)$. 
		
Because $G_v(u)$ is an external added tree, there is a leaf $x'\neq v$ in $G_v(x)$; then $m_\phi(x') = 1$. Let $y'$ be the vertex such that $\phi(x') = \phi(y')$. Now it follows, by Lemma \ref{lem:morphism-degree}, that $\deg(x') = \deg(y')$, thus $y'$ is a leaf. Since $x'$ is an added vertex, it also follows that $\mathscr{C}_y = 0$. 
Since every leaf $G$ was incident to a constraint, we conclude that $y'$ is added to $G$. It follows that $G'\backslash \{y', x'\}$ is a refinement of $G$ and that $\phi'\colon G'\backslash\{y',x'\} \to T\backslash\{\phi(y')\}$ is a suitable morphism of degree 2. This yields a contradiction with the minimality of $G'$. 
		
We conclude that $\sgon(G) \geq 3$. 
\end{proof}

Let $G$ be a graph in $\mathcal{G}_2^s$, that is connected in the sense of black and green edges. We will show that we can apply a rule to $G$. 
	
We define the graphs $H_1$, $H_2$ and $H_3$ as a single vertex, a vertex with a green loop and two vertices connected by a green edge respectively. These are exactly the graph that can be reduced to the empty graph by Rules \ref{rul:sgonEnd1}, \ref{rul:sgonEnd2} and \ref{rul:sgonEnd3}.

\begin{lemma} \label{lem:sgonComplete}
Given a non-empty connected graph $G\in \mathcal{G}_2^s$ there is a rule in $\mathcal{R}^s$ that can be applied to $G$.
\end{lemma}

\begin{proof} 
Let $G$ be a non-empty connected graph with $\sgon(G) \leq 2$. Suppose that no rule can be applied to $G$. 

We first say something about the structure of $G$. If there is a double edge between two vertices $u$ and $v$, then removing these two edges yields a disconnected graph, otherwise we could apply Rule \ref{rul:twoEdgesAndPath}. Let $u_1v_1, \ldots, u_kv_k$ be all double edges in $G$. Let $G_{i,1}, G_{i,2}$ be connected components after removing the edges $u_iv_i$. If there is a degree 2 vertex with a green loop, then removing this vertex yields a disconnected graph, otherwise we could apply rule \ref{rul:degree2GreenLoop}. Let $v_1, \ldots, v_l$ be all degree two vertices with a green loop. Let $G'_{i,1}, G'_{i,2}$ be the connected components after removing $v_i$. Let $H$ be the element of \[\{G_{i,j} \mid 1\leq i\leq k, j\in \{1,2\} \} \cup \{G'_{i,j} \mid 1\leq i\leq l, j\in\{1,2\}\}\] with the minimum number of vertices. Notice that there is at most one vertex $v$ in $H$ with $\deg_{H}(v) \neq \deg_{G}(v)$, namely the vertex that was adjacent to the removed edge or edges. Now we can say the following about $H$.  
\begin{itemize}
	\item If $H$ contains only one vertex, then we could have applied Rule \ref{rul:leaf}, \ref{rul:leafGreenLoop}, \ref{rul:degree2}, \ref{rul:degree2GreenLoop}, \ref{rul:sgonEnd1} or \ref{rul:sgonEnd2}. Thus $H$ contains at least 2 vertices. 
	\item If there is a vertex that is incident to more than one green edge, then $\sgon(G) \geq 3$ by Lemma \ref{prop:multipleGreenEdges}. So we can assume that no vertex is incident to more than one green edge. 
	\item If $H$ contains a vertex $u \neq v$ of degree 0, then $\deg_G(u) = \deg_{H}(u) = 0$. We see that $\mathscr{C}_u = \{(u,w)\}$ with $u\neq w$, because $H$ is connected in black and green edges and contains at least two vertices. By Lemma \ref{prop:greenEdgeNotSameDegree} it follows that $\deg_G(u) = \deg_G(w) = 0$. Since $G$ is connected it follows that $G = H_3$, so we can apply Rule \ref{rul:sgonEnd3}. This yields a contradiction. So we can assume that $G$ does not contain vertices with degree 0. 
	\item If $G$ contains a leaf $u \neq v$, then $\deg_{H}(u) = \deg_G(u) = 1$. We see that $u$ is incident to a green edge $uw$, with $\deg_G(w) \neq 1$, otherwise we could apply Rule \ref{rul:leaf}, \ref{rul:leafGreenLoop} or \ref{rul:twoleaves}. By Lemma \ref{prop:greenEdgeNotSameDegree}, it follows that $\sgon(G) \geq 3$. So we can assume that $G$ does not contain leaves. 
	\item If $G$ contains a vertex $u \neq v$ of degree 2, then we see that $\mathscr{C}_u = \{(u,w)\}$ with $u\neq w$, by Rules \ref{rul:degree2} and \ref{rul:degree2GreenLoop} and by the choice of $H$.
    \item We see that $H$ does not contain black loops because of Rule \ref{rul:loops}. 
	\item By Rules \ref{rul:parallelToGreenEdge} and \ref{rul:twoEdgesAndPath} and by the choice of $H$ it follows that $H$ has no multiple edges. 
\end{itemize}

Write $H'$ for the graph obtained from $H$ by removing all green loops and coloring all green edges of $H$ black. Altogether we see that $H'$ is a simple graph with at least two vertices and every vertex, except at most one, has degree at least 3. It follows that $H'$ has treewidth at least 3. 

If we change the color of all green edges to black, we see that all rules are deletions of vertices or edges, contractions of edges and/or additions of loops. Since the set of graphs with treewidth at most $k$ is closed under these operations, we see that $\tw(G) \geq \tw(H') \geq 3$. But then it follows that $\sgon(G) \geq \tw(G) \geq 3$. This yields a contradiction.  

We conclude that there is a rule in $\mathcal{R}^s$ that can be applied to $G$. 
\end{proof}

Now we have everything required to prove our main theorem:

\begin{proof}[Proof of Theorem \ref{thm:sgonSafeComplete}]
By Lemma \ref{lem:sgonSafe} we have that $\mathcal{R}^s$ is safe. By Lemma \ref{lem:sgonComplete} we can keep applying rules from $\mathcal{R}^s$ to any graph $G\in \mathcal{G}_2^s$ as long as $G$ has not been turned into the empty graph yet. It remains to prove that we can apply rules from $\mathcal{R}^s$ only a finite number of times. 

Consider the following potential function $f$:
let $f(G) = n+2m+g$ for a graph $G$ with $n$ vertices, $m$ (black) edges, and $g$ green edges. Each rule decreases $f(G)$ by at least one. Thus rules from $\mathcal{R}^s$ can only be applied a finite number of times. When no more rules can be applied, it follows that the graph has been reduced to the empty graph. Therefore $\mathcal{R}^s$ is complete. 
\end{proof}

So, we can use this set of rules to recognize graphs with stable gonality at most 2.

	\section{Reduction Rules for Stable Divisorial Gonality}
    \label{section:rulessdg}

	In this section we show a set of reduction rules to decide whether stable divisorial gonality is at most two. This set is similar to the set of rules for stable gonality, and it uses the concept of constraints of divisorial gonality.

	\subsection*{Contraints for Stable Divisorial Gonality}
	
	We use the notion of constraints as in the section about divisorial gonality. We will refer to them as red edges. We call an effective divisor of degree 2 that satisfies all conditions given by the constraints a \emph{suitable} divisor. Again, let $\mathcal{G}_2^{sd}$ be the set of all graphs with constraints with stable divisorial gonality at most 2.

	\subsection*{Reduction rules}
	\label{subsec:sdgonRules}
	
	We will now state all rules. 
      When a rule adds a red edge $uv$, and there already exists such a red edge, then the set of constraints does not change. 
	The reduction rules for stable divisorial hyperelliptic graphs are almost the same as the rules for stable hyperelliptic graphs. Instead of green edges we use red edges, and we replace Rule \ref{rul:degree2} and \ref{rul:loops} by new Rules \ref{rul:sdgonDegree2a}, \ref{rul:sdgonDegree2b} and \ref{rul:sdgonLoops}, see Figure \ref{fig:sdgonRules} for the new rules. 
	
	\begin{figure}[t]
		\centering
		\begin{tabular}{rlr}
\multicolumn{1}{l}{\small{Rule \ref{rul:sdgonDegree2a}}} & & \multicolumn{1}{l}{\small{Rule \ref{rul:sdgonLoops}}} \\
\begin{tikzpicture}
	\cirkelsEnPijl
	\lus{rood}{4.25}{0}
    \lus{white}{1.25}{0}
	\node[vertexx] (w) at (1.25,0) {};
	\node[vertexx] (v) at (0,0) {};
	\draw[edge] (v) to [relative, out=40, in=140] (w);
	\draw[edge] (w) to [relative, out=40, in=140] (v);
	\node[vertexx] (u') at (4.25,0) {};
\end{tikzpicture}  & & \begin{tikzpicture}
	\cirkelsEnPijl
	\lus{edge}{1.25}{0}
	\node[vertexx] (u) at (1.25,0) {};
	\node[vertexx] (u') at (4.25,0) {};
\end{tikzpicture} 
		\end{tabular}
		\caption{The reduction rules different from the rules for $\sgon$}  
		\label{fig:sdgonRules}
	\end{figure}

	\begin{customrule}{$\bm{T_1^{sd}}$}[=\ref{rul:leaf}] \label{rul:sdgonLeaf}
		Let $v$ be a leaf with $\mathscr{C}_v = \emptyset$. Let $u$ be the neighbor of $v$. Contract the edge $uv$. 
	\end{customrule}
	
	\begin{customrule}{$\bm{T_2^{sd}}$}[=\ref{rul:leafGreenLoop}] \label{rul:sdgonLeafGreenLoop}
		Let $v$ be a leaf with $\mathscr{C}_v = \{(v,v)\}$. Let $u$ be the neighbor of $v$. Contract the edge $uv$. 
	\end{customrule}
	
	\begin{customrule}{$\bm{S_{1a}^{sd}}$} \label{rul:sdgonDegree2a}
    Let $v$ be a vertex of degree 2 with $\mathscr{C}_v = \emptyset$. 
	Let $u$ be the only one neighbor of $v$. Remove $v$ and add a red loop to $u$. 
	\end{customrule}
    
    \begin{customrule}{$\bm{S_{1b}^{sd}}$}\label{rul:sdgonDegree2b}
    Let $v$ be a vertex of degree 2 with $\mathscr{C}_v = \emptyset$. Let $u_1$ and $u_2$ be the two neighbors of $v$, with $u_1 \neq u_2$. Contract the edge $u_1v$. 
	\end{customrule}
	
	\begin{customrule}{$\bm{T_3^{sd}}$}[=\ref{rul:twoleaves}] \label{rul:sdgonTwoleaves}
    Let $G$ be a graph where every leaf and every degree 2 vertex is incident to a red edge. Let $v_1$ and $v_2$ be two leaves that are connected by a red edge. Let $u_1$ and $u_2$ be their neighbors.  Contract the edges $u_1v_1$ and $u_2v_2$. 
	\end{customrule}
	
	\begin{customrule}{$\bm{S_2^{sd}}$}[=\ref{rul:degree2GreenLoop}] \label{rul:sdgonDegree2GreenLoop}
    Let $G$ be a graph where every leaf and every degree 2 vertex is incident to a red edge. Let $v$ be a vertex of degree 2 with a red loop, such that there exists a path from $v$ to $v$ in the black and red graph $G$. Let $u_1$ and $u_2$ be the neighbors of $v$. Remove $v$ and connect $u_1$ and $u_2$ with a red edge. 
	\end{customrule}

	\begin{customrule}{$\bm{L^{sd}}$} \label{rul:sdgonLoops}
		Let $v$ be a vertex with a loop. Remove all loops from $v$.  
	\end{customrule}
	
	\begin{customrule}{$\bm{P_1^{sd}}$}[=\ref{rul:parallelToGreenEdge}] \label{rul:sdgonParallelToGreenEdge}
		Let $uv$ be an edge. Suppose that there also exists a red edge from $u$ to $v$. Remove the black edge $uv$.   
	\end{customrule}
	
	\begin{customrule}{$\bm{P_2^{sd}}$}[=\ref{rul:twoEdgesAndPath}] \label{rul:sdgonTwoEdgesAndPath}
		Let $u,v$ be vertices, such that $|E(u,v)|>1$. Let $e$ and $f$ be two of those edges. If there exists another path, possibly containing red edges, from $u$ to $v$, then remove $e$ and $f$ and add a red edge from $u$ to $v$. 
	\end{customrule}

\begin{customrule}{$\bm{E_1^{sd}}$}[=\ref{rul:sgonEnd1}] \label{rul:sdgonEnd1}
Let $G$ be the graph consisting of a single vertex $v$ with $\mathscr{C}_v = \emptyset$. Remove $v$. 
\end{customrule}

\begin{customrule}{$\bm{E_2^{sd}}$}[=\ref{rul:sgonEnd2}] \label{rul:sdgonEnd2}
Let $G$ be the graph consisting of a single vertex $v$ with a red loop. Remove $v$. 
\end{customrule}

\begin{customrule}{$\bm{E_3^{sd}}$}[=\ref{rul:sgonEnd3}] \label{rul:sdgonEnd3}
Let $G$ be the graph consisting of a two vertices $u$ and $v$ that are connected by a red edge. Remove $u$ and $v$. 
\end{customrule}

We will write $\mathcal{R}^{sd}$ for the set of these reduction rules. We can now state our main theorem. 

\begin{theorem} \label{thm:sdgonSafeComplete}
The set of rules $\mathcal{R}^{sd}$ is safe and complete for $\mathcal{G}_2^{sd}$. 
\end{theorem}

	\subsection*{Safeness}
	
	We will show that the set $\mathcal{R}^{sd}$ is safe for $\mathcal{G}_2^{sd}$.

	\begin{lemma} \label{lem:sdgonFirstRule}
		Rule \ref{rul:sdgonLeaf} is safe. 
	\end{lemma}
	
	\begin{proof}
    Let $v$ be the vertex in $G$ to which the rule is applied.
    
		Suppose that $\sdgon(G) \leq 2$. Since every refinement of $G$ is a refinement of $H$, it is clear that $\sdgon(H) \leq 2$. 
		
		Suppose that $\sdgon(H) \leq 2$. Then there exists a refinement $H'$ of $H$ and a suitable divisor $D$. Write $u$ for the neighbor of $v$ in $G$. There exists a divisor $D'\sim_\mathscr{C} D$ such that $D'(u) \geq 1$. Now define $G'$ as $H'$ with a leaf $v$ added to $u$. Look at the divisor $D'$ on $G'$. For every vertex $w \in H'$ we can reach a divisor with one chip on $w$. By adding $v$ to every firing set that contains $u$, we see that we can still reach $w$ in $G'$. And we can reach a divisor with a chip on $v$ by firing $H'$ in $G'$. Thus $D'$ is a suitable divisor on $G'$. We conclude that $\sdgon(G) \leq 2$. 
	\end{proof}

	\begin{lemma}
		Rule \ref{rul:sdgonLeafGreenLoop} is safe.
	\end{lemma}
	
	\begin{proof}
		This proof is analogous to the proof of Lemma \ref{lem:sdgonFirstRule}.   
	\end{proof}

	\begin{lemma}
		\ref{rul:sdgonDegree2a} is safe. 
	\end{lemma}
	
	\begin{proof}
        Let $v$ be the vertex in $G$ to which the rule is applied.
    
		Let $u$ be the neighbor of $v$. Suppose that $\sdgon(G) \leq 2$. Let $G'$ be a refinement of $G$ such that there exists a suitable divisor $D$. Let $C$ be the cycle through $v$ and $u$. Notice that $D$ is equivalent to a divisor $D'$ with two chips on $C$. If $G\backslash\{C\}\cup\{u\}$ is a tree, then we are done. Otherwise we see that $D'$ is equivalent with $D''$, where $D''(u) = 2$. Let $H'$ be $G'\backslash C \cup \{u\}$ with a red loop at $u$. We see that $H'$ is a refinement of $H$ and $D''$ is a suitable divisor for $H'$, thus $\sdgon(H) \leq 2$.   
		
		Suppose that $\sdgon(H) \leq 2$. Then there exists a refinement $H'$ of $H$ such that $D$, with $D(u) = 2$, is a suitable divisor. Let $G'$ be $H'$ without the red loop on $u$ and with a vertex $v$ with two edges to $u$. Then $G'$ is a refinement of $G$. It is clear that $D$ is a suitable divisor for $G'$ too. We conclude that $\sdgon(G) \leq 2$. 
	\end{proof}

    \begin{lemma}
		\ref{rul:sdgonDegree2b} is safe. 
	\end{lemma}

	\begin{proof}
        Let $v$ be the vertex in $G$ to which the rule is applied.

		Let $u_1, u_2$ be the neighbors of $v$. We know that $u_1\neq u_2$. Suppose that $\sdgon(G) \leq 2$. Then there is a refinement $G'$ of $G$ such that there exists a suitable divisor on $G'$. Notice that $G'$ is a refinement of $H$ too, so $\sdgon(H) \leq 2$. 
		
		Now suppose that $\sdgon(H) \leq 2$. Let $H'$ be a refinement of $H$ and $D$ a suitable divisor on $H'$. If the edge $u_1u_2$ is subdivided in $H'$, then $H'$ is a refinement of $G$ too, and we are done. Assume that the edge $u_1u_2$ is not subdivided in $H'$. Let $D'$ be the divisor with $D'(u_1) = D'(u_2) = 1$. If $D \sim_\mathscr{C} D'$, then we can subdivide $u_1u_2$ to obtain a refinement $G'$ of $G$. By starting with the divisor $D'$, we can reach all vertices of $H'$ as before and we can reach $v$ by firing all vertices of $H'$. 
		
		Suppose that $D \nsim_\mathscr{C} D'$. There exist divisors $D_{u_1}$ and $D_{u_2}$ such that $D_{u_1} \sim_\mathscr{C} D \sim_\mathscr{C} D_{u_2}$ and $D_{u_1}(u_1) = 1$ and $D_{u_2}(u_2) = 1$. It follows that $D_{u_1}(u_2) = 0$ and $D_{u_2}(u_1) = 0$. Let $A_1, \ldots A_k$ be the level set decomposition of the transformation from $D_{u_1}$ to $D_{u_2}$. Let $D_j$ be the divisor before firing $A_j$. Let $D_i$ be the first divisor such that $D_i(u_2) > 0$, then it is clear that $u_2 \notin A_j$ for $j < i$. Since $D_i \nsim_\mathscr{C} D'$ we see that $D_i(u_1) = 0$. It follows that $u_1\in A_{i-1}$. We conclude that there is a chip fired along the edge $u_1u_2$. For every vertex $w$ in $H'$ we can find a divisor $D_w$ with at least one chip on $w$, let $B_{w,1}, \ldots B_{w,l_w}$ be all sets that occur in the level set decomposition of the transformation from $D_{u_1}$ to $D_w$. For all $w$, $i$, let $E_{w,i}$ be the set of all edges along which a chip is fired by the set $B_{w,i}$. Subdivide all edges that occur in some $E_{w,i}$ to obtain a refinement $G'$ of $G$. Let $V_{w,i}$ be set of vertices that are added on the edges in $E_{w,i}$. Define $B'_{w,i}$ as $B_{w,i}$ together with all added vertices that are on an edge with both endpoints in $B_{w,i}$. We can replace every set $B_{w,i}$ by two sets $B'_{w,i}, B'_{w,i}\cup V_{w,i}$ in the level set decomposition $B_{w,1}, \ldots B_{w,l_w}$ to see that for every vertex $w$ we can still reach a divisor with at least one chip on $w$. Notice that we still satisfy all constraints.  And for every vertex in $V_{w,i}$ we will encounter a divisor with a chip on that vertex when we transform $D_{u_1}$ in $D_w$. We conclude that $\sdgon(G) \leq 2$. 
	\end{proof}   
  
\begin{lemma}\label{thm:sdgonTwoleaves}
\ref{rul:sdgonTwoleaves} is safe. 
\end{lemma}

\begin{proof}
    Let $v_1$ and $v_2$ be the vertices in $G$ to which the rule is applied.

Suppose that $\sdgon(G) \leq 2$. Let $G'$ be a minimum refinement of $G$ and $D$ a suitable divisor on $G'$. Let $u_1$ and $u_2$ be the neighbors of $v_1$ and $v_2$ in $G$. We distinguish three cases. 
	
Case 1: Suppose that $u_1 \neq u_2$, and that there does not exist a path from $v_1$ to $v_2$ in black and red edges, except the red edge $v_1v_2$. Then we can reach all vertices in $G_{v_1}(u_1)$ with only one chip, thus $G_{v_1}(u_1)$ is a black tree. Thus $G_{v_1}(u_1)$ contains a leaf that is not incident to a red edge. This yields a contradiction with the minimality of $G'$. 
	
Case 2: Suppose that $u_1 \neq u_2$ and that there exists a path $P$, possibly containing red edges, from $v_1$ to $v_2$. Assume that $D\nsim D'$ where $D'$ is the divisor such that $D'(u_1) = D'(u_2) = 1$. Let $a_0 = v_1, a_1, \ldots, a_k=u_1$ be the added vertices on the edge $v_1u_1$ and let $b_0 = v_2, b_1, \ldots, b_l=u_2$ be the added vertices on the edge $v_2u_2$. Assume without loss of generality that $k < l$. It is clear that all vertices $a_0, \ldots, a_k, b_0, \ldots, b_l$ lie on $P$. Notice that firing the sets $\{a_i, b_i \mid i\leq j\}$ for $j = 0, 1, \ldots, k$ results in the divisor $D_k$ with $D_k(a_k) = D_k(b_k) = 1$. Thus $D_k(u_1) = 1$. Since $b_k$ is an internal added vertex and $G'$ is a minimum refinement, we see that $\deg(b_k) = 2$. 

Suppose that $u_1$ is incident to a red edge $u_1x$. We know that $x\neq b_k$, since $b_k$ is an added vertex. Let $D''$ be the divisor with $D''(u_1) = D''(x) = 1$. Let $A_1, \ldots, A_s$ be the level set decomposition of the transformation from $D_k$ to $D''$. We see that $u_1$ cannot lose its chip. Thus $b_k$ fires its chip to one of its neighbors when we fire $A_1$. But then we see that the cut of $A_1$ is at least two, and we can only fire one chip. This yields a contradiction. We conclude that $u_1$ is not incident to a red edge. 

By the conditions of the rule it follows that $\deg(u_1) \geq 3$. Let $w\notin P$ be a neighbor of $u_1$. Now we see that $G'_{u_1}(w)$ is a black tree. It follows that $G'_{u_1}(w)$ contains a leaf that is not incident to a red edge. Since $G'$ is a minimum refinement, this yields a contradiction. Altogether we conclude that $k=l$. 

Let $P_1, P_2$ be the two arcs of $P$ between $u_1$ and $u_2$. Notice that, if there are two chips on $P$, then they are either on $u_1$ and $u_2$ or on the same arc $P_i$. Suppose that there are divisors $E, E'$ such that $E \sim_\mathscr{C} E'$ and that there is a set $A$ in the level set decomposition of $E'- E$ such that $u_1 \in A$ and $u_2 \notin A$. It follows that there is a chip fired along each of the arcs $P_1$ and $P_2$. This yields a contradiction. We conclude that for every firing set it holds that either $u_1$ and $u_2$ are both fired or they are both not fired. 

Now let $H'$ be $G'$ without the red edge $v_1v_2$ and with a red edge $u_1u_2$. We see that $D$ is a suitable divisor for $H'$ as well. Thus $\sdgon(H) \leq 2$. 

Case 3: Suppose that $u_1 = u_2$. This case is analogous to case 2. 

Suppose that $\sdgon(H) \leq 2$. Then it is clear that $\sdgon(G) \leq 2$. 
\end{proof}

	\begin{lemma}
		Rule \ref{rul:sdgonDegree2GreenLoop} is safe.
	\end{lemma}
	
	\begin{proof}
	This proof is analogous to the proof of cases two and three in the proof of Lemma \ref{thm:sdgonTwoleaves}, so we omit it. 
	\end{proof}
		\begin{lemma}
		Rule \ref{rul:sdgonLoops} is safe. 
	\end{lemma}
	
	\begin{proof}
		There will never be a chip fired over a loop, so loops do nothing for the stable divisorial gonality. Thus $\sdgon(G) \leq 2$ if and only if $\sdgon(H) \leq 2$. 
	\end{proof}
	
	\begin{lemma} \label{thm:sdgonParallelToGreenEdge}
		Rule \ref{rul:sdgonParallelToGreenEdge} is safe. 
	\end{lemma}
	
	\begin{proof}
        Let $uv$ be the edge in $G$ to which the rule is applied.

		Suppose that $\sdgon(G) \leq 2$. Let $G'$ be a refinement of $G$ and $D$ a suitable divisor with $D(u) = D(v) = 1$. Let $G_{uv}$ be all internal and external added vertices to the edge $uv$. Now define $H' = G'\backslash G_{uv}$. Look at the divisor $D$ on $H'$ and notice that $D$ is a suitable divisor. Observe that $H'$ is a refinement of $H$. Thus $\sdgon(H) \leq 2$. 
		
		Suppose that $\sdgon(H) \leq 2$. Let $H'$ be a refinement of $H$ and $D$ a suitable divisor. Add an edge $uv$ to $H'$, to obtain a refinement $G'$ of $G$. We see that $D'$ is a suitable divisor for $G'$, thus $\sdgon(G) \leq 2$. 
	\end{proof}
		
	\begin{lemma} 
		Rule \ref{rul:sdgonTwoEdgesAndPath} is safe. 
	\end{lemma}
	
	\begin{proof}
    This proof is analogous to the proof of Lemma \ref{lem:ruleC3safe}. 
	\end{proof}

    \begin{lemma} \label{lem:sdgonEnd1} \label{lem:sdgonEnd2} \label{lem:sdgonEnd3} \label{lem:sdgonLastRule}
Rules \ref{rul:sdgonEnd1}, \ref{rul:sdgonEnd2} and \ref{rul:sdgonEnd3} are safe. 
\end{lemma}

\begin{proof}
All those graphs have stable gonality at most 2, so the statement holds true. 
\end{proof}

Now we have proven that all rules are safe, so we can conclude the following.
\begin{lemma}\label{lem:sdgonSafe}
The set of rules $\mathcal{R}^{sd}$ is safe for $\mathcal{G}_2^{sd}$. \qed
\end{lemma}

\subsection*{Completeness}
	
Now we will prove that $\mathcal{R}^{sd}$ is complete for $\mathcal{G}_2^{sd}$, i.e.\ if $G\in \mathcal{G}_2^{sd}$, then $G$ can be reduced to the empty graph by applying finitely many rules. 

\begin{lemma} \label{lem:sdgon-connected}
	Let $G$ and $H$ be graphs. If $G$ is connected in the sense of black and red edges, and $H$ can be produced from $G$ by applying some rules, then $H$ is connected in the sense of black and red edges.
\end{lemma}
\begin{proof}
    See the proof of Lemma \ref{lem:sgon-connected}. 
\end{proof}

\begin{lemma} \label{prop:multipleRedEdges}
	Let $G$ be a graph. If there is a vertex $v$ with $|\mathscr{C}_v| > 1$, then $\sdgon(G) \geq 3$. 
\end{lemma}
	
\begin{proof}
	See the proof of Lemma \ref{ConflictingLabelsLemma}. 
\end{proof}

\begin{lemma}
	Let $G$ be a graph where every leaf is incident to a red edge, so if $\deg(u) = 1$ then $|\mathscr{C}_u| > 0$ for all $u$. Suppose that $u$ is a leaf and $(u,v)$ is a red edge. If $\deg(v)\neq 1$, then $\sdgon(G) \geq 3$. 
\end{lemma}

\begin{proof}
	Analogous to the proof of Lemma \ref{DegreeOneLemma}.
\end{proof}

We define the graphs $H_1$, $H_2$ and $H_3$ as a single vertex, a vertex with a red loop and two vertices connected by a red edge respectively, these are the graphs that can be reduced to the empty graph by rules \ref{rul:sdgonEnd1}, \ref{rul:sdgonEnd2} and \ref{rul:sdgonEnd3}.

\begin{lemma} \label{lem:sdgonComplete}
Given a non-empty connected graph $G\in \mathcal{G}_2^{sd}$ there is a rule in $\mathcal{R}^{sd}$ that can be applied to $G$.
\end{lemma}

\begin{proof} 
Let $G$ be a connected graph with $\sdgon(G) \leq 2$, and suppose that no rule can be applied to $G$. 

As in the proof of Theorem \ref{lem:sgonComplete}, if there are two edges between the vertices $u$ and $v$, then removing these edges lead to $G$ being disconnected. And if there is a degree 2 vertex $v$ with a red loop, then removing $v$ yields a disconnected graph. Let $H$ be the smallest connected component that can be created by removing two parallel edges or a degree 2 vertex, as in the proof of Theorem \ref{lem:sgonComplete}. 

Now we color all red edges black and remove all loops to obtain $H'$, as in the proof of Theorem \ref{lem:sgonComplete}. Then we see that $H'$ is a simple graph that contains at least two vertices and all vertices, except at most one, have degree at least three. Thus $H'$ has treewidth at least three. It follows that $\sdgon(G) \geq \tw(G) \geq \tw(H') \geq 3$. 

We conclude that a rule can be applied to $G$. 
\end{proof}

\begin{proof}[Proof of Theorem \ref{thm:sdgonSafeComplete}]
Lemma \ref{lem:sdgonSafe} shows that the set of reduction rules $\mathcal{R}^{sd}$ is safe. We need to show that $\mathcal{R}^{sd}$ is complete as well. Lemma \ref{lem:sdgonComplete} shows that to any graph $G\in \mathcal{G}_2^{sd}$ we can keep applying rules from $\mathcal{R}^s$ to as long as $G$ has not been turned into the empty graph yet. We can use the same potential function as in the proof of \ref{thm:sgonSafeComplete} to prove that we can apply rules from $\mathcal{R}^{sd}$ only a finite number of times. Thus $\mathcal{R}^{sd}$ is complete. 
\end{proof}

Thus, we can use the set $\mathcal{R}^{sd}$ to recognize graphs with stable divisorial gonality at most 2.

\section{Algorithms}
\label{section:algorithms}
In this section, we discuss how the reduction rules of Sections~\ref{divSection}, \ref{section:stablegonalityrules} and
\ref{section:rulessdg} lead to efficient algorithms that recognize graphs with divisorial gonality, stable gonality, or stable divisorial gonality $2$ or lower. 

\begin{theorem}[=Theorem \ref{thm:hoofdstellingIntro}]
There are algorithms that, given a graph $G$ with $n$ vertices and $m$ edges, decide whether $\dgon(G)\leq 2$, $\sgon(G) \leq 2$ or $\sdgon(G)\leq 2$ in $O(n \log n+m)$ time.
\label{thm:time}
\end{theorem}

For each of the algorithms, we first check whether a given graph is connected or not. The only disconnected graphs with $\dgon(G)\leq 2$, $\sgon(G) \leq 2$ or $\sdgon(G)\leq 2$ are forests that consist of two trees. This can be done in linear time. After that we can assume that our graph is connected, and the algorithms will use the reduction rules to check whether $\dgon(G)\leq 2$, $\sgon(G) \leq 2$ or $\sdgon(G)\leq 2$. 
Each of the algorithms is of the following form: first, with an additional rule, we can ensure that on each pair of vertices, there are at most two parallel edges and at most one red or green edge; then, we verify that the treewidth of $G$ is at most $2$; if not, we can directly decide negatively. After that, we repeatedly apply a rule, until none is possible. By Theorems~\ref{mainTheorem}, \ref{thm:sgonSafeComplete}, and 
\ref{thm:sdgonSafeComplete}, the final graph after all rule applications is empty, if and only if  $\dgon(G)\leq 2$, $\sgon(G) \leq 2$ or
$\sdgon(G)\leq 2$, respectively. It is not hard to see that for each rule,  deciding whether the rule can be applied, and if so, applying it, can be done in polynomial time; the fact that we work with graphs of
treewidth $2$ yields an implementation with $O(n \log n)$ steps, after the $O(m)$ work to remove the extra parallel edges.

\begin{proof}[Proof of Theorem \ref{thm:time}] We now present the details of the algorithm, in various steps. \\

\noindent \textsc{Step 1.} First we introduce a new rule for each of the gonalities, cf.\ Figure \ref{fig:extraRule}. 
\begin{figure}
	\centering
	\begin{tabular}{rlr}
		\multicolumn{1}{l}{\small{Rule \ref{rul:sgonMultipleEdges}}} & & \multicolumn{1}{l}{\small{Rule \ref{rul:dgonMultipleEdges}}}  \\
		\begin{tikzpicture}
	\cirkelsEnPijl
	\node[vertexx] (u) at (1.25,.375) {};
	\node[vertexx] (v) at (1.25,-.375) {};
	\node[vertexx] (u') at (4.25,0.375) {};
	\node[vertexx] (v') at (4.25,-.375) {};
	\draw[edge] (u) -- (v);
	\draw[edge] (u) to [relative, out=50, in=130] (v);
	\draw[edge] (v) to [relative, out=50, in=130] (u);
	\draw[groen] (u') -- (v');
    \node (l) at (-0.75,0) {};
\end{tikzpicture}  & & \begin{tikzpicture}
	\cirkelsEnPijl
	\node[vertexx] (u) at (1.25,.375) {};
	\node[vertexx] (v) at (1.25,-.375) {};
	\node[vertexx] (u') at (4.25,0.375) {};
	\node[vertexx] (v') at (4.25,-.375) {};
	\draw[edge] (u) -- (v);
	\draw[edge] (u) to [relative, out=50, in=130] (v);
	\draw[edge] (v) to [relative, out=50, in=130] (u);
	\draw[rood] (u') -- (v');
    \node (l) at (-0.75,0) {};
\end{tikzpicture} 
\end{tabular}
\caption{An extra reduction rule for stable gonality and one for divisorial gonality}  
\label{fig:extraRule}
\end{figure}
For divisorial gonality we introduce the following rule: 
\begin{customrule}{$\bm{M^d}$} \label{rul:dgonMultipleEdges}
Let $u,v$ be vertices, such that $|E(u,v)| \geq 3$. Let $k=\left \lfloor{(|E(u,v)|-1)/2}\right \rfloor$, then remove $2k$ edges between $u$ and $v$ and add a constraint $(u,v)$.
\end{customrule}
Note that this rule merely shortcuts repeated applications of Rule \ref{rul:dc3}.
Now we introduce a new rule for stable gonality.
\begin{customrule}{$\bm{M^s}$} \label{rul:sgonMultipleEdges}
Let $u,v$ be vertices, such that $|E(u,v)| \geq 3$. Remove all edges in $E(u,v)$ and add a green edge from $u$ to $v$.
\end{customrule}
It is clear that this rule is the same as first applying Rule \ref{rul:twoEdgesAndPath} and then applying Rule \ref{rul:parallelToGreenEdge} to all remaining edges $(u,v)$. 
For stable divisorial gonality we introduce a similar rule. 
\begin{customrule}{$\bm{M^{sd}}$} \label{rul:sdgonMultipleEdges}
Let $u,v$ be vertices, such that $|E(u,v)| \geq 3$. Remove all edges in $E(u,v)$ and add a red edge from $u$ to $v$.
\end{customrule}
This is again the same as first applying Rule \ref{rul:sdgonTwoEdgesAndPath} and then applying Rule \ref{rul:sdgonParallelToGreenEdge} to all remaining edges $(u,v)$.

All applications of these rules can be done in $\mathcal{O}(m)$ at the start of the algorithm, after which we know that no pair of vertices can have more than two edges between them. For stable and stable divisorial gonality, by application of Rule \ref{rul:loops} and \ref{rul:sdgonLoops} we can also ensure in $\mathcal{O}(m)$ time that no loops exist  (in the case of divisorial gonality loops can be safely ignored).\\

\noindent \textsc{Step 2.} Recall that treewidth is a lower bound on divisorial gonality, stable gonality and stable divisorial gonality \cite{Gijs}. Therefore it follows that if $\tw(G)>2$, the algorithm can terminate. Checking whether treewidth is at most $2$ can be done in linear time. Hereafter, we assume our graph has treewidth at most $2$. \\

\noindent \textsc{Step 3.} Recall that graphs of treewidth $k$ and $n$ vertices have at most $kn$ edges. It follows that the underlying simple graph has at most $2n$ edges. By our previous steps there are at most $2$ edges between a pair of vertices and no loops, so there are at most $4n$ edges left.

\begin{lemma}
Let $G$ be a graph of treewidth at most $2$, with at most $2$ parallel edges between each pair of vertices. For each of the three collections of rules, the number of
successive applications of rules is bounded by $O(n)$.
\end{lemma}

\begin{proof}
For the reduction rules for divisorial gonality from Section \ref{divSection} note that all rules except Rule \ref{rul:dc3} always remove at least one vertex. It follows they can be applied at most $n$ times. For Rule \ref{rul:dc3} note that the only case where it removes no vertex is when the cycle $C$ consists of a double edge between two vertices. Since in this case we remove $2$ edges and there are at most $4n$ edges left, it follows this rule can also be applied at most $2n$ times. Therefore at most $3n$ rules can be applied before we reach the empty graph.

For the rules in Section~\ref{section:stablegonalityrules}, consider the following potential function $f$:
let $f(G) = n+2m+g$ for a graph $G$ with $n$ vertices, $m$ remaining (black) edges, and $g$ green edges. One easily observes that $f(G)=O(n)$, and each rule decreases $f(G)$ by at least one. 

The same argument holds for the collection of reduction rules of stable divisorial gonality
from Section~\ref{section:rulessdg}.
\end{proof}

Thus, for $\dgon$ and $\sgon$, the given sets of rules already lead to polynomial time algorithms: for each of the rules, one can test in polynomial time for a given graph (with green edges or constraints) if the rule can be applied to the graph, and if so, apply the rule in polynomial time. \\

\noindent \textsc{Step 4.} In the remainder of the proof, we will argue that there is an implementation that leads to recognition algorithms running in $O(n \log n+m)$ time.

For the case of divisorial gonality, by Lemma \ref{ConflictingLabelsLemma}, each vertex can only have one constraint that applies to it. Therefore, checking for compatible constraints can be done in time linear in the number of vertices that will be removed by the rule. To maintain this property it is necessary to check for conflicts whenever a rule adds a new constraint, but this can be done in constant time per rule application.

By keeping track of the degree of each vertex as we apply rules, vertices with degree one or zero can be found in constant time. 
It is also easy to
directly detect when the graph is one induced cycle; e.g., by keeping track of the number of vertices of degree unequal two or with an incident selfloop or colored edge; this takes care of Rule~\ref{rul:dc1}. The remaining problem then is efficiently finding applications of the following rules: \[ (\ast) \ \  \mbox{ \ref{rul:dc2}, \ref{rul:dc3}, \ref{rul:degree2GreenLoop}, \ref{rul:twoEdgesAndPath}, \ref{rul:sdgonDegree2GreenLoop} and \ref{rul:sdgonTwoEdgesAndPath}.} \]

To do this in $O(n\log n)$ time in total, we use a technique, also employed by Bodlaender et al.~\cite[Section 6]{BodlaenderDDFLP16}. 
(See also \cite{Hagerup00} for more results on dynamic algorithms on graphs of bounded treewidth.) Due to the highly technical aspects, the discussion here is not
self-contained, and assumes a knowledge of techniques for monadic second order logic formulas on graphs of bounded treewidth.

We build a data structure that allows the following operations: deletions of vertices, deletions of edges, contractions of edges, adding a (possible colored) self-loop to a vertex changing the color of an edge (e.g., turning an edge into a green edge), and
deciding whether the rules $(\ast)$ 
can be applied, and if so, yielding the pair of
vertices to which the rule can be applied.

First, by \cite[Lemma 2.2]{BodlaenderH98}, we can build in $O(n)$ time a tree decomposition of $G$
of width $8$, such that the tree $T$ in the tree decomposition is binary and has $O(\log n)$ depth.
We augment $G$ with a number of labels for edges and vertices. Vertices are labelled with a value that
can be `selfloop', `deleted', or `usual' (no selfloop, not deleted). Edges are labelled with a value that can be one of the following: `usual' (black edge, without parallel edge); `parallel'; `red'; `green';
`deleted'; `contracted'. When we perform an operation that changes the multiplicity of an edge, deletes it, or changes its color, we do not change the tree decomposition, but instead only change the label of
the edge.

For each of the rules in $(\ast)$, 
there is a sentence $\phi$ in Monadic Second Order Logic (MSOL) with two free vertices variables, such that $\phi(v,w)$ holds if and only if the corresponding rule in $(\ast)$ is applicable to vertices $v$ and $w$. 

We can modify these sentences $\phi$ to sentences $\phi'$ that apply to graphs where edges can be labelled with labels `deleted' or `contracted', i.e., if $G'$ is the graph obtained from $G$ by deleting edges with the label `deleted', then $\phi(v,w)$ holds for $G'$ if and only if $\phi'(v,w)$ holds for $G$.

First, when we perform a quantification over edges, we add a condition that the edge is not a deleted edge. A quantification $\forall F\subseteq E: \psi(F)$ becomes \[\forall F\subseteq E: (\forall e\in F: \neg \mathrm{deleted}(e)) \Rightarrow \psi(F);\]
a quantification $\exists F\subseteq E: \psi(F)$ becomes
\[\exists F\subseteq E: (\forall e\in F: \neg \mathrm{deleted}(e)) \wedge \psi(F).\]
Secondly, we modify the sentence to deal with contracted edges, by making the following three changes. For any quantification over sets of edges, we ensure that there are no edges with the label contracted. For any quantification over sets of vertices, we add the condition that for any edge with the label `contracted', either both endpoints are in the set or both endpoints of the contracted edge are not in the set. Finally, elementary predicates like ``$v=w$" or
$v$ is incident to $e$, become a phrase in MSOL: ``$v=w$" is translated to an MSOL sentence that expresses that there is a path from $v$ to $w$ (possibly empty) with all edges on the path labeled as `contracted', and ``$v$ incident to $e$" is translated to a property that expresses that there is a path with contracted edges from $v$ to an endpoint of $e$.
In the same way, we can further modify the sentences to also deal with vertices with the label `deleted'.

A consequence of Courcelle's theorem \cite{Courcelle90} is that for a sentence $\phi(v,w)$ with two free vertex variables, we have an algorithm that,
given a tree decomposition of a graph $G$ of bounded width, determines if there are vertices $v$ and $w$ for which $\phi(v,w)$ holds, and if so, outputs the vertices $v$ and $w$, and that uses linear time.
Using the techniques from \cite{BodlaenderDDFLP16}, we can modify this algorithm such that we can update the graph by changing labels of vertices and edges, and do these queries, such that each update and query costs time that is linear in the depth of the tree of the tree decomposition, i.e., $O(\log n)$. (For the details, we refer to \cite[Section 6]{BodlaenderDDFLP16}.)\\

We are now ready to wrap up. Each of the rules can be executed by doing $O(1)$ deletions of vertices, edges, contractions of edges, adding or deleting a selfloop, or changing color or multiplicity of an edge. 
As discussed above, each of such operation costs $O(\log n)$ time to the data structure, and finding (if it exists) a pair of vertices to which we can apply 
a rule from $(\ast)$ also costs
$O(\log n)$ time. It follows that the total time is $O(\log n)$ times the number of times we apply a rule; as argued earlier in this section, we have $O(n)$ rule applications, giving a total of $O(n \log n)$ time, excluding the $O(m)$ time to remove edges with multiplicity larger than 2. This gives Theorem~\ref{thm:time}.
\end{proof}

\begin{remark}
The constant factors produced by the heavy machinery of Courcelle's theorem and  tree decompositions of width $8$ are very large; a simpler algorithm  with a larger asymptotic (but still polynomial) algorithm will be faster in practice.
\end{remark}

\begin{corollary}[=Corollary \ref{cor:automorfismeIntro}] \label{cor:automorfisme}
There is an algorithm that, given a two-edge-connected graph $G$, decides in $O(n \log n + m)$ time whether $G$ admits an involution $\sigma$ such that the quotient $G/\langle \sigma \rangle$ is a tree. \qed
\end{corollary}
\begin{proof}
If $G$ is two-edge-connected and $b_1(G) \leq 1$, then $G$ is a cycle. For a cycle such an involution does exist. So we can decide this in linear time. 

If $b_1(G) \geq 2$, this follows from Theorem \ref{thm:time}, since Baker and Norine \cite[Thm. 5.12]{Baker09} have shown that a two-edge-connected graph $G$ with $b_1(G)\geq 2$ is divisorial hyperelliptic precisely if $G$ admits an automorphism as stated in the theorem. 
\end{proof}

\section{Conclusion}
We have provided an explicit set of safe and complete reduction rules for (multi-)graphs of  divisorial, stable, and stable divisorial gonality at most $2$, and we have shown that stable, divisorial, and stable divisorial hyperelliptic graphs can be recognized in polynomial time (actually, $O(n\log n+m)$). 

Finally, we mention some interesting open questions on (divisorial) gonality from the point of view of algorithmic complexity: (a) Can hyperelliptic graphs be recognized in linear time? (b) Does there exist a set of reduction rules to recognize graphs with gonality 3? (c) Is gonality fixed parameter tractable? (d) Which problems become fixed parameter tractable with gonality as parameter? (e) Is there an analogue of Courcelle's theorem for bounded gonality?

\end{document}